\documentclass[a4paper, UKenglish, cleveref, autoref, thm-restate]{lipics-v2021}

\usepackage{cite}
\usepackage{graphicx}
\usepackage{textcomp}
\usepackage{xcolor}

\usepackage{xspace}
\usepackage{bbm}
\usepackage{comment}
\usepackage{amsmath}
\usepackage{booktabs}
\usepackage{amsfonts}
\usepackage{amssymb}
\usepackage{amsthm}
\usepackage{thmtools,thm-restate}
\usepackage[textsize=small]{todonotes}
\usepackage{tikz}
\usetikzlibrary{arrows,arrows.meta,automata,positioning} 
\usepackage{stmaryrd}
\usepackage{mathdots}
\usepackage{algorithm}
\usepackage[noend]{algpseudocode}
\usepackage{thmtools}
\usepackage{lineno}
\usepackage{enumerate}
\usepackage{hyperref}

\usepackage{url}
\usepackage[bibliography=common]{apxproof}
\newcommand{\kamb}[0]{$k$-ambiguous }
\newcommand{\support}[0] {
  \textup{support}}

\newcommand{\ignore}[1]{}

\newcommand{\C}{\mathcal{C}}

\newcommand{\N}{\mathbb{N}}
\newcommand{\Z}{\mathbb{Z}}

\newcommand{\set}[1]{\{#1\}}
\newcommand{\setof}[2]{\set{#1 \mid #2}}

\newcommand{\tupple}[1]{\langle #1 \rangle}

\newcommand{\reach}{\textsc{Reach}}

\newcommand{\Wlog}{W.l.o.g.\xspace}
\newcommand{\mywlog}{w.l.o.g.\xspace}

\newcommand{\dett}[0] {\textup{Det}\xspace}
\newcommand{\hist}[0] {\textup{Hist}\xspace}

\newcommand{\ndet}[0] {\textup{NonDet}}

\newcommand{\amb}[0] {\textup{BAmb}\xspace}
\newcommand{\xamb}[1] {\textup{#1-Amb}\xspace}

\newcommand{\xndet}[1]{\textup{#1-}\ndet}

\newcommand{\trans}[1]{\stackrel{#1}{\longrightarrow}}
\newcommand{\widebar}[1]{\overline{#1}}

\newcommand{\suff}{\textup{suff}}

\newcommand{\eff}{\textsc{eff}}

\newcommand{\Nplus}{\N_{\geq 0}}

\newcommand{\nctwo}{\textsc{NC$^2$}\xspace}
\newcommand{\ptime}{\textsc{PTime}\xspace}

\newcommand{\pspace}{\textsc{PSpace}\xspace}
\newcommand{\exptime}{\textsc{ExpTime}\xspace}
\newcommand{\expspace}{\textsc{ExpSpace}\xspace}

\newcommand{\ackermann}{\textsc{Ackermann}\xspace}

\newcommand{\myparagraph}[1]{\vskip 0.1 cm \noindent \textbf{#1.}}

\newcommand{\inc}{\textup{inc}}
\newcommand{\dec}{\textup{dec}}
\newcommand{\ztest}{\textup{ztest}}
\newcommand{\ski}{\textup{skip}}

\usepackage{tikz}
\usetikzlibrary{arrows.meta}
\usetikzlibrary{automata, calc}
\usetikzlibrary{decorations.pathmorphing}
\tikzstyle{accepting}=[path picture={%
  \draw let 
    \p1 = (path picture bounding box.east),
    \p2 = (path picture bounding box.center)
    in
      (\p2) circle (\x1 - \x2 - 2pt);
  }]

\title{Languages of Boundedly-Ambiguous Vector Addition Systems with States} %TODO Please add

\author{Wojciech Czerwi\'nski}
{University of Warsaw}
{wczerwin@mimuw.edu.pl}
{https://orcid.org/0000-0002-6169-868X}
{}%Supported by the ERC grant INFSYS, agreement no. 950398.

\author{Łukasz Orlikowski}
{University of Warsaw}
{l.orlikowski@mimuw.edu.pl}
{https://orcid.org/0009-0001-4727-2068}
{}%Supported by the ERC grant INFSYS, agreement no. 950398.

\authorrunning{W. Czerwiński and Ł. Orlikowski} %TODO mandatory. First: Use abbreviated first/middle names. Second (only in severe cases): Use first author plus 'et al.'

\Copyright{Wojciech Czerwiński and Łukasz Orlikowski} %TODO mandatory, please use full first names. LIPIcs license is "CC-BY";  http://creativecommons.org/licenses/by/3.0/


 %TODO mandatory; please add comma-separated list of keywords

\category{} %optional, e.g. invited paper

\relatedversion{} %optional, e.g. full version hosted on arXiv, HAL, or other respository/website
\funding{Both authors are supported by the ERC grant INFSYS, agreement no. 950398.}%optional, to capture a funding statement, which applies to all authors. Please enter author specific funding statements as fifth argument of the \author macro.

\nolinenumbers %uncomment to disable line numbering

\ArticleNo{13}
\begin{document}

\maketitle

\begin{abstract}
The aim of this paper is to deliver broad understanding of a class of languages of boundedly-ambiguous VASSs, that is $k$-ambiguous VASSs for some natural $k$. These are languages of Vector Addition Systems with States with the acceptance condition defined by the set of accepting states such that each accepted word has at most $k$ accepting runs. We develop tools for proving that a given language is not accepted by any $k$-ambiguous VASS. Using them we show a few negative results: lack of some closure properties of languages of $k$-ambiguous VASSs and undecidability of the $k$-ambiguity problem, namely the question whether a given VASS language is a language of some $k$-ambiguous VASS. In fact we show an even more general undecidability result stating that for any class containing all regular languages and only $k$-ambiguous VASS languages for some $k \in \N$ it is undecidable whether a language of a given $1$-dimensional VASS belongs to this class. Finally, we show that the regularity problem is decidable for $k$-ambiguous VASSs.
\end{abstract}

% !TEX root = main.tex
\section{Introduction}
Determinism is a central notion in computer science.
Deterministic systems often allow for more efficient algorithms.
On the other hand, usually deterministic systems are less expressive, so in many cases for a non-deterministic system
there is no equivalent deterministic system and one cannot use more efficient techniques.
For those reasons, there is recently a lot of research devoted to various notions restricting nondeterminism
in a milder way than determinism. The hope is that systems having the considered properties are more expressive than
the deterministic ones, but still allow for robust algorithms design.
One prominent example of such a notion is unambiguity; a system is unambiguous if
for every word there is at most one accepting run over this word. In the last decade unambiguous systems
were intensely studied and for various classes of infinite-state systems the unambiguous case
turns out to be much more tractable~\cite{DBLP:conf/dcfs/Colcombet15,DBLP:conf/icalp/Raskin18,DBLP:conf/stacs/MottetQ19,DBLP:conf/concur/CzerwinskiFH20,DBLP:conf/icalp/GoosK022}.
Similar notions were also investigated recently, like $k$-ambiguity (each word is accepted by at most $k$ runs)
and history-determinism (a weakened version of determinism),
in both cases one can design more efficient algorithms in some cases~\cite{DBLP:conf/concur/CzerwinskiH22,history-deterministic}.

In this paper we focus on studying milder version of determinism for Vector Addition Systems with States (VASS).
VASSs and related Petri nets are popular and fundamental models of concurrency with many applications
both in theory and in practical modelling~\cite[Section 5]{DBLP:journals/siglog/Schmitz16}.
Languages of VASSs with restricted nondeterminism were already studied for several years, mostly with the acceptance condition
being the set of accepting states.
In~\cite{DBLP:conf/concur/CzerwinskiFH20} it was shown that the universality
problem for unambiguous VASSs is decidable in \expspace, in contrast to \ackermann-completeness of the problem
for VASSs without that restriction. In~\cite{DBLP:conf/concur/CzerwinskiH22,DBLP:journals/lmcs/CzerwinskiH25}
the language equivalence problem was considered for unambiguous VASS and more generally for $k$-ambiguous VASSs
and it was shown to be decidable and \ackermann-complete,
in contrast to undecidability in general~\cite{DBLP:journals/tcs/ArakiK76},
even in dimension one~\cite[Thm. 20]{DBLP:conf/lics/HofmanMT13}.
The choice of universality and equivalence problems is deliberate here. The complexity of the seemingly more natural
language emptiness problem (or equivalently of the reachability problem) does not depend on the unambiguity assumption.
Indeed, one can always transform a system to a deterministic one by assigning all the transitions unique labels
without affecting nonemptiness.
Therefore, in order to observe a difference in complexity or decidability after restricting nondeterminism,
one should consider problems, which do not simply ask about the existence of some object. Some natural problems concerning languages
of that type are the universality problem and the language equivalence problem and these two problems
were extensively studied for many models with the unambiguity assumption or some modification of it.
First of all, the universality and equivalence problems are solvable in \ptime~\cite{DBLP:journals/siamcomp/StearnsH85}
(and even in \nctwo~\cite{DBLP:journals/ipl/Tzeng96}) for unambiguous finite automata (UFA),
in contrast to \pspace-hardness of both problems for NFA.
The universality problem was also shown decidable for unambiguous register automata in a sequence
of papers~\cite{DBLP:conf/stacs/MottetQ19,DBLP:conf/stacs/BarloyC21,DBLP:conf/icalp/CzerwinskiMQ21}
with improving complexity,
which contrasts undecidability in the case without restricted nondeterminism~\cite[Thm 5.1]{DBLP:journals/tocl/NevenSV04}.
This line of research culminated in a very elegant contribution~\cite{DBLP:conf/lics/BojanczykKM21,DBLP:journals/theoretics/BojanczykFKM24}, which showed in particular that not only the universality problem, but also the equivalence
problem is solvable in \exptime for unambiguous register automata, and in \ptime in the case of fixed number of registers.
Another popular restriction of nondeterminism was also studied recently: history-determinism.
A system is history-deterministic if its nondeterminism
can be resolved on the fly, based on the history of a particular run. The equivalence and inclusion problems were shown
to be decidable for one-dimensional history-deterministic VASSs~\cite{DBLP:conf/fossacs/PrakashT23},
but undecidable for two-dimensional history-deterministic VASSs~\cite{history-deterministic}.

Despite a lot of research on algorithms for unambiguous and boundedly-ambiguous ($k$-ambiguous for some $k \in \N$) VASS
not much is known about the class of languages recognisable by such models.
In particular, to our best knowledge, till now it was even not known
whether there exist any VASS language (we consider acceptance by states), which is not a language of an unambiguous VASS.
The reason behind this lack of knowledge was absence of any technique, which can show that a given language of a VASS cannot be
recognised by an unambiguous VASS. The quest for such a technique is natural and the question deserves investigation.
Analogous problems were considered for other models of computation.
For finite automata the question trivialises, as deterministic automata recognise all the
regular languages. However, already for weighted automata over a field the problems are highly nontrivial and were recently
studied in-depth~\cite{DBLP:conf/lics/BellS23,DBLP:journals/corr/abs-2410-03444}, in particular it is decidable
whether a given weighted automaton is unambiguisable, so equivalent to some unambiguous weighted automaton~\cite{DBLP:conf/lics/BellS23}. The problem whether a given context-free language is unambiguisable
is known to be undecidable since the 60s~\cite{undecidable_cfl,DBLP:journals/mst/Greibach68}.
The aim of this paper is deepening the understanding of the class of languages recognisable by boundedly-ambiguous VASS
and its subclasses.

\myparagraph{Our contribution}
In the paper we deliver several results, which help understanding unambiguous, $k$-ambiguous and boundedly-ambiguous VASS languages. Our first main tool is Lemma~\ref{lem:first_tool}, which delivers the first example of a VASS language, which is known
to be not a $k$-ambiguous VASS language. The second main tool is Lemma~\ref{lem:image}, which formulates a condition,
which needs to be satisfied by all $k$-ambiguous VASS languages. Using these two lemmas we can rather straightforwardly
inspect closure properties of $k$-ambiguous VASS languages in Lemma~\ref{lem:closure_prop} and show several expressivity
results in Section~\ref{subsec:expressivity}.
Further building on Lemma~\ref{lem:first_tool} we obtain our main contribution.

\begin{theorem}\label{thm:main}
For any class $\C$ of languages containing all regular languages and contained in the class of all boundedly-ambiguous VASS languages
it is undecidable to check whether the language of a given $1$-VASS accepting by states belongs to $\C$.
\end{theorem}

Consequences of Theorem~\ref{thm:main} are broad.
It reproves undecidability of regularity of $1$-dimensional VASSs
(in short $1$-VASSs)~\cite[Section 8]{DBLP:journals/lmcs/CzerwinskiL19} and
undecidability of determinisability of $1$-VASSs considered in~\cite{DBLP:conf/concur/AlmagorY22}.
It is important to emphasise that Theorem~\ref{thm:main} gives us extensive flexibility wrt. the undecidability results.
For example in~\cite{DBLP:conf/concur/AlmagorY22} it was shown that for a given $1$-VASS it is undecidable whether there is
an equivalent deterministic $1$-VASS. One can argue that possibly asking about equivalent deterministic VASS, without
a bounding dimension, is a more natural question, which deserves independent research. 
Theorem~\ref{thm:main} answers negatively all questions of that kind in one shot.
We formulate below its corollary to illustrate variety of consequences we obtain.
In particular we know now that for many classical restrictions of nondeterminism it is undecidable
whether for a given $1$-VASS there exists some equivalent one with nondeterminism restricted in that way.

\begin{corollary}\label{cor:main}
It is undecidable whether the language of a given $1$-VASS accepting by states is recognisable by some
\begin{itemize}
  \item unambiguous VASS,
  \item $k$-ambiguous VASS for given $k \in \N$,
  \item boundedly ambiguous VASS,
  \item deterministic VASS,
  \item $k$-ambiguous $1$-VASS for given $k \in \N$.
\end{itemize}
\end{corollary}

Our last main contribution is Theorem~\ref{thm:regularity}, which states that the regularity problem
is decidable for boundedly-ambiguous VASSs,
which contrasts undecidability without that assumption~\cite{DBLP:journals/tcs/ArakiK76,DBLP:conf/lics/HofmanMT13}.
One can see this result as an intuitive indication that boundedly-ambiguous VASSs are closer to deterministic VASSs
rather than to general nondeterministic VASSs.

\myparagraph{Organisation of the paper}
In Section~\ref{sec:prelim} we introduce preliminary notions and recall useful lemmas.
Section~\ref{sec:techniques} is devoted to showing the two main technical tools, namely Lemmas~\ref{lem:first_tool}~and~\ref{lem:image}.
In Section~\ref{sec:properties} we present closure properties and expressivity results for the classes of $k$-ambiguous VASSs.
Next, in Section~\ref{sec:undecidability} we show our main result, namely Theorem~\ref{thm:main}.
Theorem~\ref{thm:regularity} is proved in Section~\ref{sec:regularity}.
Finally, in Section~\ref{sec:future} we discuss interesting future research directions.

% !TEX root = main.tex
\section{Preliminaries}\label{sec:prelim}
\myparagraph{Basic notions}
For $a,b \in \N$ such that $a \leq b$ we write $[a, b]$ for the set of integers $\{a, a+1, \ldots, b\}$.
For $a \in \N$ we write $[a]$ for the set $[1, a]$.
For a vector $v \in \Z^d$ we write $v_i$ for the $i$-th entry of $v$.
For a vector $v \in \Z^d$ we write $\support(v)$ for the set $\{i \mid v_i > 0\}$.
For two vectors $u, v \in \N^d$ we write $u \leq v$ if for all $i \in [d]$ we have $u_i \leq v_i$.
We also extend this order to $(\N \cup \{\omega\})^d$ where $\omega$ is bigger than any natural number.
We write $\N_{\omega}$ for the set $\N \cup \{\omega\}$.
Whenever we speak about the norm of vector $v \in \Z^d$ we mean $||v||=\max_{1 \leq i \leq d}|v_i|$.
For a word $w$ we denote by $\#_a(w)$ the number of letters $a$ in the word $w$.

\myparagraph{Downward-closed sets}
A set $S \subseteq N^d$ is downward-closed if for each $u, v \in \N^d$
it holds that $u \in S$ and $v \leq u$ implies $v \in S$. For $u \in \N^d$ we write $u\downarrow$ for the set  $\setof{v}{v \leq u}$.
If a downward-closed set is of the form $u\downarrow$ we call it a
down-atom. Observe, that a one-dimensional set $S \subseteq \N$ is downward-closed if either $S = \N$
or $S = [0, n]$ for some $n \in \N$. Thus we have either $S = \omega \downarrow$ in the first case or $S = n \downarrow$ in the second case.
So equivalently a downward-closed set $D \subseteq \N^d$ is a down-atom if it 
is of a form $D = D_1 \times \ldots \times D_d$, where for all $i \in [d]$ we have that $D_i$
is a downward-closed one dimensional set. 
For simplicity we write $D = (n_1, n_2, \ldots , n_d)\downarrow$ if $D = n_1 \downarrow \times n_2 \downarrow \ldots n_d \downarrow$. Therefore each down-atom is of the form $u \downarrow$ where $u \in (\N \cup \omega)^d$.
The following proposition will be helpful in our considerations.
\begin{proposition}[Lemma 17 in \cite{CzerwinskiLMMKS18}, \cite{Dickson}]\label{prop:down_atom}
    Each downward-closed set in $\N^d$ is a finite union of down-atoms. 
\end{proposition}

\myparagraph{Vector Addition Systems with States}
A $d$-dimensional Vector Addition System with States ($d$-VASS) is a nondeterministic finite automaton with $d$
non-negative integer counters. Transitions of the VASS manipulate these counters. Formally, we define VASS $V$ as  $V = (\Sigma, Q, \delta, I, F)$ where $\Sigma$ is a finite alphabet, $Q$ is a finite set of automaton states, $\delta \subseteq Q \times (\Sigma \cup \varepsilon) \times \Z^d \times Q$ is a transition relation, $I \subseteq Q \times \N^d$ is a finite set of initial configurations and $F \subseteq Q \times \N^d$ is a finite set of final configurations.
For a transition $t = (s, a, v, s') \in \delta$ we say that the transition is over $a$ or the transition reads letter $a$. We also write $\eff(t)$ for the effect of transition, which is $v$. We define the norm of transition $t$ as the norm of $v$.

A $d$-VASS can be seen as an infinite-state labelled transition system in which each configuration is a pair $(s, u) \in Q \times \N^d$. We denote such configuration as $s(u)$. We define the norm of the configuration as the norm of $u$. A transition $t = (s, a, v, s') \in \delta$ can be fired in a configuration $q(u)$ if and only if $q=s$ and $u' \geq 0$ where $u'=u + v$. After firing the transition configuration is changed to $s'(u')$. We also define run as a sequence of transitions, which can be fired one after another from some configuration. For a run $\rho=t_1t_2\ldots t_n$ we write $\eff(\rho)$ for $\Sigma_{i=1}^n \eff(t_i)$. If we want to say something only about $j$th (for $j \in [d]$) entry of the $\eff(\rho)$ we write $\eff_j(\rho)$. We also define $\support(\rho)$ as $\support(\eff(\rho))$. We say that a run $\rho$ is from configuration $q(u)$ to $p(u')$ if the sequence of transitions can be fired from $q(u)$ and the final configuration is $p(u')$. We say that a run is a loop if the state of its initial configuration is the same as the state of the final configuration. For two runs $\rho_1 = \alpha_1\ldots \alpha_n$ and $\rho_2 = \beta_1 \ldots \beta_k$ if the sequence $\alpha_1\ldots \alpha_n\beta_1 \ldots \beta_k$ is a run $\rho$ then we can write $\rho = \rho_1\rho_2$. We say, that a run $\rho=t_1\ldots t_n$ is over $w=\lambda_1\ldots \lambda_n \in (\Sigma \cup \{\varepsilon\})^*$ (or reads $w$) if and only if for each $i \in [n]$ transition $t_i$ is over $\lambda_i$. We denote by $w(\rho)$ the word read by $\rho$. We say that the length of a run $\rho$ is equal to $n$ if it consists of $n$ transitions. For the length of a run, we write $|\rho|$.

We say, that VASS is $\varepsilon$-free if there is no transition over $\varepsilon$. In this work, unless stated otherwise, we work with $\varepsilon$-free VASSs.
\myparagraph{Languages of VASSs}
A run of a VASS is accepting if it starts in an initial configuration $c_0 \in I$ and ends in an accepting configuration. 
A configuration is accepting if it covers some configuration from the set of final configurations $F$.
In other words configuration $q(v)$ is accepting if there exists a configuration $q(v') \in F$ such that $v \geq v'$.
For a VASS $V$ we define its language as the set of all words read by accepting runs and we denote it by $L(V)$.
Languages defined in this way are called coverability languages.
Sometimes we consider languages with other acceptance condition (reachability languages), in that case we indicate it clearly.

We sometimes consider reachability languages in which run ending in configuration $q(v)$ is accepting if and only if $q(v) \in F$.
Sometimes, we consider VASSs, where the set of accepting configurations is infinite, but has a specific form. For instance we consider downward-VASSs, where the set of accepting configurations is possibly infinite, but it is downward-closed. In this type of VASSs a run ending in configuration $q(v)$ if and only if $q(v) \in F$ where $F$ is a downward-closed set of accepting configurations. 
We say, that VASS $V$ is deterministic if it has only one initial configuration, it is $\varepsilon$-free and for each state $q$ and each letter $a \in \Sigma$ there is at most one transition over $a$ leaving the state $q$.
We say, that VASS $V$ is $k$-ambiguous if and only if for every $w \in L(V)$ we have at most $k$ accepting runs over $w$. We also say then, that $L(V)$ is a $k$-ambiguous language. 
For a $d$-VASS consider a function $r : (Q \times \N^d \times \delta)^* \times (Q \times \N^d) \times \Sigma \rightarrow \delta^*$ that, given a history of the run (configurations and taken transitions), current configuration $q(v)$ and a next letter $\lambda \in \Sigma$,
returns a sequence of possibly many transitions. One of these transitions is over $\lambda$, all the other are over $\varepsilon$,
and the sequence can be fired from the current configuration $q(v)$. 
Let us call $r$ a resolver. We say, that $d$-VASS $V$ is history-deterministic if and only if it has one initial configuration and there exists a resolver $r$ such that for each $w \in L(V)$ run $\rho$ over $w$ from the initial configuration given by the resolver is accepting.

We denote by $\dett$, $\hist$, $\xamb{k}$, $\amb$ and $\ndet$ the class of languages of respectively deterministic, history-deterministic, \kamb, all boundedly-ambiguous and fully nondeterministic VASSs languages.
We sometimes denote by \xamb{1} the class of unambiguous VASSs languages.
By $\xndet{d}$ we denote the class of languages of $d$-VASSs.

\myparagraph{Well-quasi order}
A quasi-order $\preceq$ defined on set $X$ is a relation satisfying:
\begin{itemize}
       \item For each $x \in X$ it holds $x \preceq x$ (reflexivity)
       \item For each $a,b,c \in X$ we have that $a \preceq b$ and $b \preceq c$ implies $a \preceq c$ (transitivity)
\end{itemize}
We say, that a quasi-order is a well quasi-order (WQO) if and only if there is no infinite antichain or infinite decreasing sequence. For simplicity, we say that pair $(X, \preceq)$ is a WQO if and only if $\preceq$ is a WQO defined on $X$. We say that $x_1, x_2 \in X$ are incomparable if and only if $x_1 \npreceq x_2$ and $x_2 \npreceq x_1$. Below we present a commonly known lemma, which presents useful equivalent conditions for the relation to be a WQO.
For the proof see for instance \cite{schmitz2012algorithmic}.
\begin{lemma}\label{lem:equiv_cond_wqo}
   For each $(X, \preceq)$ such that $\preceq$ is a quasi-order defined on $X$ the following conditions are equivalent:
   \begin{enumerate}
       \item $(X, \preceq)$ is a WQO.\label{cond:wqo}
       \item Every infinite sequence $x_1, x_2, x_3, \ldots$ such that $x_i \in X$ contains infinite non-decreasing subsequence $x_{n_0} \preceq x_{n_1} \preceq x_{n_2}\preceq \ldots$ (with $n_0 < n_1 < n_2 < \ldots$).\label{cond:infi}
       \item In every infinite sequence $x_1, x_2, x_3, \ldots$ such that $x_i \in X$ one can find $i < j$ such that $x_i \preceq x_j$.\label{cond:domi}
   \end{enumerate}
\end{lemma}
Using Lemma \ref{lem:equiv_cond_wqo} we prove, that one can define WQO on configurations of a VASS (extended with $\omega$-coordinates): 
\begin{lemma}\label{wqo}
   For every $d \in \N$ and finite set $Q$ we have, that $\preceq$ defined on $Q \times \N_{\omega}^d$ as for each $q_1, q_2 \in Q$ and $v, v' \in \N_{\omega}^d$ we have $(q_1, v) \preceq  (q_2, v') \iff v \leq v' \land q_1=q_2$ is a well-quasi-order .
\end{lemma}
\begin{proof}[Proof of Lemma \ref{wqo}]
   From the definition, it is clear that $\preceq$ is transitive and reflexive. We prove, that it is a WQO. For this we need Dickson's Lemma \cite{schmitz2012algorithmic}.
   \begin{lemma}
       (Dickson’s Lemma). Let $(A, \preceq_1)$ and $(B, \preceq_2)$ be two WQOs. Then $(A \times B, \preceq)$ is a WQO where $\preceq$ is defined for each $a, a' \in A$ and $b, b' \in B$ as $(a, b) \preceq (a', b') \iff a \preceq_1 a' \land b \preceq_2 b'$.
   \end{lemma}
   Let us fix $d \in \N$. Observe, that $(Q, =)$ and $(\N_{\omega}, \leq)$ are WQOs. Hence, by applying Dickson's Lemma $d$ times we get, that $(Q \times \N_{\omega}^d, \preceq)$ where for each $q_1, q_2 \in Q$ and $v, v' \in \N_{\omega}^d$ we have $(q_1, v) \preceq  (q_2, v') \iff \forall_{i \in [d]} v_i \leq v_i' \land q_1=q_2$ is a WQO.
\end{proof}
% !TEX root = main.tex
\section{Tools for separating \amb and \ndet}\label{sec:techniques}
In this section we develop two techniques for showing that a language is not recognised by a $k$-ambiguous VASS: Lemma \ref{lem:first_tool} and Lemma \ref{lem:image}.

\subsection{Dominating block}
The first tool is based on the following observation about a language not recognised by a \kamb VASS.
\begin{lemma}\label{lem:first_tool_simple}
    For every $k \in \N_+$ the language 
    $$L_k = \{a^{n_1}ba^{n_2}ba^{n_3}b\ldots a^{n_{k+2}} \mid  \exists_{i \in [k+1]} \ n_i \geq n_{i+1} \}$$ 
    is not recognised by a \kamb VASS. 
\end{lemma}
For showing various properties of boundedly-ambiguous VASS Lemma \ref{lem:first_tool_simple} is sufficient.
In one case, for proving Theorem \ref{thm:main}, we need a strengthening of Lemma \ref{lem:first_tool_simple}, which is presented in the following lemma:
\begin{lemma}\label{lem:first_tool}
    Let $\Sigma$ be an alphabet such that $b \notin \Sigma$ and let $L$ be a language over $\Sigma$. For each function $f: L \rightarrow \N_\omega$ such that $\sup{f} = \omega$ (recall that $\sup{f} = \sup{\{f(x) \mid x \in L\}}$)
    language $L_1 = \{a^{n_1}ba^{n_2}ba^{n_3}b\ldots a^{n_{k+2}}bw \mid w \in L, \exists_{i \in [k+1]} \ n_i \geq n_{i+1} \lor n_{k+2} \geq f(w) \}$ is not recognised by a \kamb VASS.
\end{lemma}

Before proving Lemma \ref{lem:first_tool} we show how it implies Lemma \ref{lem:first_tool_simple}.
\begin{proof}[Proof of Lemma \ref{lem:first_tool_simple}]
    Let us fix $k \in \N_+$ and let $L = \set{\varepsilon}$. Let $f: L \rightarrow \N_\omega$ be defined as $f(\varepsilon) = \omega$.
    Hence by Lemma \ref{lem:first_tool} we get that $L_k$ is not recognised by a \kamb VASS.
\end{proof}

Below we sketch the idea behind the proof of Lemma~\ref{lem:first_tool}. The whole proof can be found in the Appendix.
We assume, for the sake of contradiction, that $L_1$ is recognised by a $k$-ambiguous VASS and aim at a contradiction by showing $k+1$ different runs over the same word. We first consider $k+1$ words $w_1, \ldots, w_{k+1} \in L_1$,
where $u \in L$ is a particularly chosen word and $N_0 < N_1 < \ldots < N_{k+2} \in \N$ are particularly chosen constants:
\[
w_i = a^{N_1!} b a^{N_2!} b \ldots a^{N_i!} b a^{N_0!} b a^{N_{i+2}!} b a^{N_{i+3}!} b \ldots a^{N_{k+2}!} b u.
\]
Then we dive into combinatorics of VASS runs $\rho_i$ over words $w_i$ and conclude that there are specific pumping cycles
in $\rho_i$. We show two main lemmas, which analyse the structure of the runs.
The first one gives conditions for finding a specific loop in a run.
Observe, that this lemma is general and is not restricted to $k$-ambiguous VASSs.

\begin{lemma}\label{lem:loop}
For each $d$-dimensional VASS $V$, each subset of counters $S \subseteq [d]$ and each $n \in \N$ there exists a constant $M:=M(V,S, n) \geq 1$ such that in every run in $V$, starting from a configuration in which values of counters from $S$ are at most $n$, which is of length at least $M$ there exists a loop, which has non-negative effect on the counters from $S$. 
\end{lemma}

The proof of Lemma~\ref{lem:loop}, which can be found in the Appendix, uses rather standard techniques similar to the ones in the Karp-Miller construction solving the coverability problem for VASS~\cite{coverability_tree}. Lemma~\ref{lem:loop} is used to find loops in words $w_i$. For example
one can show that if $N_1$ is big enough then one can find a loop as in Lemma~\ref{lem:loop} in the first block of letters $a$.
Similarly, if $N_2$ is big enough wrt. to $N_1$, one can find an appropriate loop in the second block of letters $a$.
Using this approach one can find loops in blocks of letters $a$, if the block is much longer than the whole prefix before it.
This is however not true for the block of length $N_0$. Therefore we need a more sophisticated tool in order to be able
to modify runs over all the words $w_i$ to runs over the same word.

This requires the more challenging part of the proof, which provides a detailed characterisation of runs of $k$-ambiguous VASS.
We formulate below one of the used lemmas in order to illustrate the kind of arguments we consider and, as it is also used in Section~\ref{subsec:image}.

\begin{lemma}\label{lem:decomp}
    Let $L$ be a language over $\Sigma = \{a,b\}$ recognised by some \kamb $d$-VASS $V$. For each $n \in \N$ there exists a constant $C$ such that each run $\rho$ in $V$ such that:
    \begin{itemize}
        \item Run $\rho$ is a prefix of an accepting run.
        \item Run $\rho$ is reading $a^{m_1}ba^{m_2}b\ldots a^{m_{n-1}}ba^{m_n}$.
    \end{itemize}
    can be decomposed as:
    \begin{enumerate}
        \item $\rho = \alpha_1\beta_1^{a_1}\alpha_1'\alpha_2\beta_2^{a_2}\alpha_2'\ldots\alpha_{n}\beta_{n}^{a_{n}}\alpha_{n}'$ for some $a_1, a_2, \ldots, a_{n} \geq 1$ and for each $i \in [1,n]$ we have $|\alpha_i|, |\beta_i|, |\alpha_i'| \leq C$.\label{cond:first:lem:decomp}
        \item For each $j < n$ we have  $w(\alpha_j') \in L(a^*b)$ and $w(\alpha_n') \in L(a^*)$ \label{cond:two:lem:decomp}
        \item For each $j \in [n]$ we have that $w(\alpha_j), w(\beta_j) \in L(a^*)$ \label{cond:three:lem:decomp}
        \item For each $j \in [n]$ we have that $\beta_j$ is either a loop or $\varepsilon$. Moreover if $m_j \geq 2 \cdot C+1$ then $\beta_j \neq \varepsilon$.\label{cond:four:lem:decomp}
        \item For each $j \in [n]$ let $A_j = \bigcup_{1 \leq i < j} \support(\beta_i)$ then $\beta_j$ is nonnegative on counters from $[d] \setminus A_j$\label{cond:five:lem:decomp}
        \item For each $j \in [n]$ there is no $\lambda$, $\delta$ and $\lambda'$ such that $\delta$ is a nonnegative loop on counters from $[d] \setminus A_j$, $\lambda' \neq \varepsilon$ and $\lambda\delta\lambda' = \alpha_j\beta_j$  \label{cond:seven:lem:decomp}
    \end{enumerate}
\end{lemma}

The proof of Lemma~\ref{lem:decomp} uses the Lemma~\ref{lem:loop} and a lot of combinatorial observations.
It is moved to the Appendix.

By appropriate use of Lemma~\ref{lem:decomp} and other auxiliary lemmas we show that $\rho_i$ can be modified a bit into
runs $\rho'_i$, which are all different, all accepting and all over the same word $w$, where
\[
w = a^{n_1} b a^{n_2} b \ldots a^{n_{k+2}} b u
\]
and $n_j = 2m^{k+3-j}\Pi_{l= j}^{k+2}N_l!$ (where $m$ is the maximal norm of a transition).
More concretely, $\rho'_i$ is obtained by pumping the loops $\beta_i$ in $\rho_i$ more times.
A challenge it to show that the resulting $\rho'_i$ are indeed different, which we achieve by more careful, but technical
investigation of the runs.
Existence of $k+1$ different runs over the same word $w$ is a contradiction
with the assumption that $L_1$ is recognised by a $k$-ambiguous VASS and finishes the proof.

\begin{toappendix}
%Our plan for the proof of Lemma \ref{lem:first_tool} is to assume, for the sake of contradiction, that $L_1$ is recognised by a $k$-ambiguous VASS and reach a contradiction by showing $k+1$ different runs over the same word. For this we will use two pumping techniques: Lemma \ref{lem:loop} and \ref{lem:decomp}.
%%Hence, we present first a lemma, which gives conditions for finding a specific loop in a run. Observe, that this lemma is general and is not restricted to $k$-ambiguous VASSs.
%%\begin{lemma}\label{lem:loop}
%%    For each $d$-dimensional VASS $V$, each subset of counters $S \subseteq [d]$ and each $n \in \N$ there exists a constant $M:=M(V,S, n) \geq 1$ such that in every run in $V$, starting from a configuration in which values of counters from $S$ are at most $n$, which is of length at least $M$ there exists a loop, which has non-negative effect on the counters from $S$. 
%%\end{lemma}
%
\subsection{Proof of Lemma \ref{lem:loop}}
   To prove Lemma \ref{lem:loop} we will need a notion of a domination tree, which is inspired by the coverability tree \cite{coverability_tree} and the main goal of it is to get the constant $M$ from it easily.
   Domination tree for $d$-VASS $V$ and initial configuration $q(u)$ such that $q$ is a state of the VASS and $u \in \N_{\omega}^d$ is a tree constructed by the following algorithm:
   \begin{enumerate}
       \item Create root of the tree, label it with $q(u)$ and mark it as "new".
       \item While a node marked as "new" exists select a node $v$ marked as "new" labelled with $q_1(u')$ and do the following:
       \begin{itemize}
           \item If on the path from the root to $v$ exists node $\bar{v}$, different than $v$, labelled with $q_1(\bar{u}')$ such, that $u' \geq \bar{u}'$ mark $v$ as "old"
           \item Otherwise for each transition $t$ enabled in $q_1(u')$ obtain configuration $q_t(u_t)$ resulting from firing transition $t$ from $q_1(u')$ and create child $v_t$ of $v$ in the tree. Label it with $q_t(u_t)$ and mark it as "new". After processing all transitions mark $v$ as "old".
       \end{itemize}
   \end{enumerate}
Formally, we extend firing a transition from configurations from the set $Q \times \N^d$ to configurations from the set $Q \times \N_\omega^d$ by setting, that for each $k \in \Z$ we have $\omega + k = \omega$. Now let us observe the following claim, which can be proven by standard techniques using well quasi orders and Kőnig's Lemma~\cite{konig}.

\begin{claim}\label{clm:domi}
   Every domination tree is finite.
\end{claim}
\begin{proof}
  Let us fix some domination tree $T$. Kőnig's Lemma \cite{konig} states that a finitely branching tree is infinite if and only if it has an infinite path.
  Hence since each vertex in $T$ has a finite degree (because VASS has a finite number of transitions) it is enough to show, that each path from the root is finite. Assume, towards contradiction, that an infinite path exists in $T$. Let $t_1, t_2, \ldots$ be consecutive labels on this path. Let $d$ be the dimension of the VASS for which $T$ was constructed and let $Q$ be the states of the VASS. Therefore $t_1, t_2, \ldots \in Q \times \N_{\omega}^d$. Recall, that $Q \times \N_{\omega}^d$ with $\preceq$ defined as $q_1(x) \preceq q_2(y) \iff \forall_{i \in [d]} x_i \leq y_i \land q_1 = q_2$ is a WQO by Lemma \ref{wqo}. Hence by Lemma \ref{lem:equiv_cond_wqo} we have, that there is a domination $t_i \preceq t_j$ for $i < j$. Hence, by a construction algorithm, node $v$, which corresponds to label $t_j$ has to be a leaf and hence the path can not be infinite, which contradicts the assumption, that the path was infinite. Hence $T$ is finite.
\end{proof}
Now we are ready to prove Lemma \ref{lem:loop}.
   For each state $q$ of $V$ we construct the domination tree from configuration $q(u)$ where $u = (u_1, u_2, \ldots, u_d)$ where for $i \in S$ we have $u_i = n$ and for $i \notin S$ we have $u_i = \omega$. As constant $M$ we take the maximal depth of such trees plus one. Observe, that we do not have to check all options $u_i = k \leq n$ (because a sequence of transitions which can be fired from some configuration can also be fired from a bigger configuration)
   and $M$ is well-defined because of the Claim \ref{clm:domi}.
   Assume, towards contradiction, that there exists a run $\rho$ of length at least $M$, starting from a configuration $q_1(v)$ in which values of counters from $S$ are at most $n$ in which there is no loop, which has a non-negative effect on counters from $S$. Because for $i \in [d]$ we have $v_i \leq u_i$ the same sequence of transitions can be fired from $q_1(u)$. Let $T$ be the domination tree constructed from $q_1(u)$. We can follow run $\rho$ in $T$ upon reaching a leaf of the tree. Let $\lambda$ be the prefix of $\rho$, which reached the leaf in the tree. Observe, that $|\lambda| < M$. By the construction of the domination tree we know, that there are two possibilities:
   \begin{enumerate}
       \item There is no transition, which can be fired from the configuration reached by $\lambda$.
       Therefore $\rho=\lambda$, which contradicts with $|\rho| \geq M$
       \item There has to be a domination in the tree.
       Therefore there exists two configurations $p(v_1)$ and $p(v_2)$ which we visited along $\lambda$ such that $p(v_1)$ was visited before $p(v_2)$ and $v_1 \leq v_2$. Let $\beta$ be the sequence of transition, which leads from $p(v_1)$ to $p(v_2)$. Observe, that because $v_1 \leq v_2$ and for each $i \in S$ $v_{1_i} \neq \omega$ and $v_{2_i} \neq \omega$ loop $\beta$  has non-negative effect on all of the counters from $S$. Recall, that $\beta$ is part of $\rho$, which contradicts the fact, that $\rho$ does not contain a loop, which is non-negative on all of the counters from $S$.
   \end{enumerate}
\end{toappendix}

\begin{toappendix}
%Beside Lemma~\ref{lem:loop} we also need another tool for characterisation of the structure of runs of $k$-ambiguous VASS,
%namely Lemma~\ref{lem:decomp}, which was formulated in Section~\ref{sec:techniques}. Here we present its proof.

\subsection{Proof of Lemma \ref{lem:decomp}}
    We prove this lemma by induction on $n$. For $n = 0$ we have $\rho=\varepsilon$ and hence all decomposition conditions are satisfied. Now we assume, that this Lemma is true for $n-1$ and we show that this implies Lemma \ref{lem:decomp} for $n$. 

    Let $\pi_1$ be the prefix of $\rho$ reading $a^{m_1}ba^{m_2}b\ldots a^{m_{n-1}}$ and let $t$ be the next transition of $\rho$ after $\pi_1$. Hence $\rho=\pi_1t\pi_2$ for some $\pi_2$. By induction assumption we can apply Lemma \ref{lem:decomp} to $\pi_1$ and get constant $C$  and decomposition:
    $$\pi_1 = \alpha_1\beta_1^{a_1}\alpha_1'\alpha_2\beta_2^{a_2}\alpha_2'\ldots\alpha_{n-1}\beta_{n-1}^{a_{n-1}}\lambda$$
    Observe, that we changed $\alpha_{n-1}'$ to $\lambda$ as we will redefine $\alpha_{n-1}'$. Now let us set $\alpha_{n-1}' = \lambda t$. Hence
    $$\rho = \alpha_1\beta_1^{a_1}\alpha_1'\alpha_2\beta_2^{a_2}\alpha_2'\ldots\alpha_{n-1}\beta_{n-1}^{a_{n-1}}\alpha_{n-1}'\pi_2$$
    Recall the definition of set $A_n$ from condition \ref{cond:five:lem:decomp}, let $m$ be the maximal norm of a transition, $s$ be the maximal norm of an initial configuration and let us define constants $T = s+2(n-1)mC$ and
    $C_n = M(V, A_n, T)$ given by Lemma \ref{lem:loop}. Let us also define constant $K_n = \max_{S \subseteq [d]} M(V,S, T+mC_n)$. Let also $C' = \max(C + 1, C_n, K_n)$. This will be a constant for Lemma \ref{lem:decomp} and $n$.
    We have two cases. The first case is that $\pi_2=\alpha_n\beta_n^{a_n}\alpha_n'$ for some loop $\beta_n$, which is nonnegative on counters from $[d] \setminus A_n$ and $|\alpha_n\beta_n| \leq C_n$ and $\beta_n$ is not a prefix of $\alpha_n'$. The second case is that such a decomposition of $\pi_2$ does not exist. We start the proof with the second case.
    \myparagraph{Case 2}
    Observe, that for each $i \in [n-1]$ and $l \in [d] \setminus A_n \subseteq [d] \setminus A_i$:
    $$\eff_l(\beta_i) = 0$$
    Hence:
    \begin{multline*}
    \eff_l(\alpha_1\beta_1^{a_1}\alpha_1'\alpha_2\beta_2^{a_2}\alpha_2'\ldots\alpha_{n-1}\beta_{n-1}^{a_{n-1}}\alpha_{n-1}') = \Sigma_{i=1}^{n-1}(\eff_l(\alpha_i)+\eff_l(\alpha_i')) \leq \\
    m \Sigma_{i=1}^{n-1}(|\alpha_i|+|\alpha_i'|) \leq 2(n-1)mC = T
    \end{multline*}
    Hence, because of the definition of $C_n$ and the fact that in this case, the decomposition of $\pi_2$ does not exist, we have that $|\pi_2| \leq C_n \leq C'$. Then we set $\alpha_n = \pi_2$ and $\beta_n = \alpha_n' = \varepsilon$.
    Then clearly from the induction assumption and the definition of $\alpha_{n-1}', \alpha_n, \beta_n, \alpha_n'$ conditions \ref{cond:first:lem:decomp}-\ref{cond:three:lem:decomp} and \ref{cond:five:lem:decomp} are satisfied.
    Condition \ref{cond:four:lem:decomp} is satisfied for $j < n$ by induction assumption and for $j=n$ we have $m_j = |\alpha_n\beta_n\alpha_n'| \leq C_n \leq C$ hence it is also satisfied.
    Condition \ref{cond:seven:lem:decomp} is satisfied for $j < n$ by induction assumption. It is satisfied for $j = n$ because otherwise $\pi_2$ could have been decomposed and we are in case 1.
    
    \myparagraph{Case 1}
    We know, that $\pi_2$ can be decomposed as described above. If there are multiple decompositions we choose any of the ones minimizing $|\alpha_n\beta_n|$. Now, we show that decomposition of $\rho$:
    $$\rho = \alpha_1\beta_1^{a_1}\alpha_1'\alpha_2\beta_2^{a_2}\alpha_2'\ldots\alpha_{n}\beta_{n}^{a_{n}}\alpha_{n}'$$
    satisfies all the conditions. We start with conditions \ref{cond:two:lem:decomp}-\ref{cond:seven:lem:decomp} as Condition \ref{cond:first:lem:decomp} is the most challenging to prove.
    \myparagraph{Condition \ref{cond:two:lem:decomp}}
    For $j < n-1$ this condition follows from the induction assumption. For $j = n-1$ we have that $\alpha_{n-1}' = \lambda t$. Recall that $t$ was a transition reading letter $b$ and from the induction assumption $\lambda$ does not read letter $b$.
    Moreover, $\alpha_n'$ is a suffix of $\pi_2$, which does not read the letter $b$.
    Hence the condition follows.
    \myparagraph{Condition \ref{cond:three:lem:decomp}}
    For $j < n$ the condition follows from the induction assumption and for $j=n$ we have that $\alpha_n\beta_n$ is a prefix of $\pi_2$, which does not read the letter $b$.
    \myparagraph{Condition \ref{cond:four:lem:decomp}}
    For $j < n$ we have this condition from induction assumption and for $j=n$ we have that $\beta_n \neq \varepsilon$ is a loop.
    \myparagraph{Condition \ref{cond:five:lem:decomp}}
    This condition for $j < n$ follows from inductions assumption and for $j = n$ follows from the fact that $\beta_n$ is a nonnegative loop on counters from $[d] \setminus A_n$ as required.
    \myparagraph{Condition \ref{cond:seven:lem:decomp}} 
    For $j < n$ we have it from induction assumption. For $j=n$ we have it from minimality of $|\alpha_n\beta_n|$.
    \myparagraph{Condition \ref{cond:first:lem:decomp}}
    Observe, that for $i \in [n-2]$ from induction assumption we have $|\alpha_i|, |\beta_i|, |\alpha_i'| \leq C \leq C'$ and moreover $|\alpha_{n-1}|, |\beta_{n-1}| \leq C \leq C'$. Moreover, from induction assumption, $|\alpha_{n-1}|, |\beta_{n-1}| \leq C \leq C'$ and $|\alpha_{n-1}'| = |\lambda t| \leq C + 1 \leq C'$. 
    Additionally $|\alpha_n\beta_n| \leq C_n \leq C'$. Therefore it is enough to prove that $|\alpha_n'| \leq K_n \leq C'$. Let
    $$S = [d] \setminus \bigcup_{i=1}^n \support(\beta_i)$$
    Observe, that for each $l \in S$ we have:
    \begin{multline*}
    \eff_l(\alpha_1\beta_1^{a_1}\alpha_1'\alpha_2\beta_2^{a_2}\alpha_2'\ldots\alpha_{n}\beta_{n}^{a_{n}}) = \eff_l(\alpha_n) + \Sigma_{i=1}^{n-1} (\eff_l(\alpha_i) +  \eff_l(\alpha_i')) \leq   \\
    mC_n + 2(n-1)mC
    \end{multline*}
    Recall that $K_n = M(V, S, T + mC_n)$ and $T = 2(n-1)mC$, so if $|\alpha_n'| > K_n$ then $\alpha_n'$ contains a nonnegative loop on counters from $S$.
    Hence it is enough to prove, that $\alpha_n'$ does not contain such loop. Assume, for the sake of contradiction, that $\alpha_n' = \lambda\gamma\lambda'$ such that $\gamma$ is a nonnegative loop on counters from $S$. We define accepting runs $\phi_1, \phi_2, \ldots, \phi_{k+1}$, which will read the same word, but will be different and hence this will create a contradiction with the fact that $V$ is \kamb VASS.
    Intuitively, due to loops $\beta_n$ and $\gamma$ we have more degrees of freedom in the last block and we can create $k+1$ different runs over the same word.
    As $\rho$ is a prefix of an accepting run we know, that there exists $\rho_1$ such that $\rho\rho_1$ is an accepting run.
    Firstly, we define runs $\psi_1$, $\psi_2$ and $\psi_3$, which will be parts of $\phi_1, \phi_2, \ldots, \phi_{k+1}$:
    $$\psi_1 = \alpha_1\beta_1^{a_1+b_1}\alpha_1'\alpha_2\beta_2^{a_2+b_2}\alpha_2'\ldots\alpha_{n}\beta_{n}^{a_{n}+b_n}$$
    $$\psi_2 = \lambda\gamma$$
    $$\psi_3 =  \lambda'\rho_1$$
    where for $i \in [n]$ we define $b_i = (k+1)m|\gamma||\beta_n| + \Sigma_{j=i+1}^n b_jm|\beta_j|$. The intuition is, that $b_i$ is chosen in such a way to compensate the possible decrease of counters because of additional executions of loops $\beta_{i+1}, \beta_{i+2}, \ldots, \beta_n$ and $\gamma$.
    Finally we define $\phi_i$
    $$\phi_i = \psi_1\beta_{n}^{i|\gamma|}\psi_2\gamma^{(k+1-i)|\beta_n|}\psi_3$$
    where for $i \in [n]$ we define $b_i = (k+1)m|\gamma||\beta_n| + \Sigma_{j=i+1}^n b_jm|\beta_j|$.
    Now it is enough to prove Claims \ref{clm:phi:cond:1:accepting}, \ref{clm:phi:cond:1:word} and \ref{clm:phi:cond:1:different}.
    \begin{claim}\label{clm:phi:cond:1:accepting}
        For each $i \in [k+1]$ we have that $\phi_i$ is an accepting run. 
    \end{claim}
    \begin{proof}
        Let us fix $i \in [k+1]$. Because $\rho\rho_1$ is an accepting run it is enough to prove that: 
        \begin{enumerate}
            \item For each $j \in [n-1]$ and $l \in [d]$ we have $$\eff_l(\alpha_1\beta_1^{a_1+b_1}\alpha_1'\alpha_2\beta_2^{a_2+b_2}\alpha_2'\ldots\alpha_{j}\beta_{j}^{a_{j}+b_{j}}) \geq \eff_l(\alpha_1\beta_1^{a_1}\alpha_1'\alpha_2\beta_2^{a_2}\alpha_2'\ldots\alpha_{j}\beta_{j}^{a_{j}})$$\label{phi:accepting:enum:1}
            \item For each $l \in [d]$ we have
            $$\eff_l(\psi_1\beta_{n}^{i|\gamma|}) \geq \eff_l(\alpha_1\beta_1^{a_1}\alpha_1'\alpha_2\beta_2^{a_2}\alpha_2'\ldots\alpha_{n}\beta_{n}^{a_{n}})$$
            \label{phi:accepting:enum:2}
            \item For each $l \in [d]$ we have
            \begin{multline*}
            \eff_l(\psi_1\beta_{n}^{b_n+i|\gamma|}\psi_2\gamma^{(k+1-i)|\beta_n|}) \geq  
            \eff_l(\alpha_1\beta_1^{a_1}\alpha_1'\alpha_2\beta_2^{a_2}\alpha_2'\ldots\alpha_{n}\beta_{n}^{a_{n}}\lambda\gamma)
            \end{multline*}\label{phi:accepting:enum:3}
        \end{enumerate}
        \myparagraph{Point \ref{phi:accepting:enum:1}}
        We prove this by induction on $j$. For $j = 1$ the inequality follows from the fact that for each $l \in [d]$ we have $\eff_l(\beta_1) \geq 0$. Now we show the induction step. If $\eff_l(\beta_j) \geq 0$ we get the inequality from the induction assumption. Hence now we assume $\eff_l(\beta_j) < 0$. Therefore there has to be $u < j$ such that $\eff_l(\beta_u) > 0$. Let $u$ be the maximal such. 
        Observe, that from induction assumption:
        \begin{multline*}
        \eff_l(\alpha_1\beta_1^{a_1+b_1}\alpha_1'\alpha_2\beta_2^{a_2+b_2}\alpha_2'\ldots\alpha_{s-1}\beta_{u-1}^{a_{u-1}+b_{u-1}}\alpha_u\beta_u^{a_u}\alpha_u'\ldots\alpha_j\beta_j^{a_j}\alpha_j') \geq \\
        \eff_l(\alpha_1\beta_1^{a_1}\alpha_1'\alpha_2\beta_2^{a_2}\alpha_2'\ldots\alpha_j\beta_j^{a_j}\alpha_j')
        \end{multline*}
        And moreover:
        $$b_u\eff_l(\beta_u) \geq b_u$$
        and for $t \in [u+1,j]$ we have:
        $$b_t\eff_l(\beta_t) \geq -b_tm|\beta_t|$$
        Hence:    
        \begin{multline*}
            \eff_l(\alpha_1\beta_1^{a_1+b_1}\alpha_1'\alpha_2\beta_2^{a_2+b_2}\alpha_2'\ldots\alpha_{j}\beta_{j}^{a_{j}+b_{j}}) = \\
            \eff_l(\alpha_1\beta_1^{a_1+b_1}\alpha_1'\alpha_2\beta_2^{a_2+b_2}\alpha_2'\ldots\alpha_{s-1}\beta_{u-1}^{a_{u-1}+b_{u-1}}\alpha_u\beta_u^{a_u}\alpha_u'\ldots\alpha_j\beta_j^{a_j}\alpha_j') + \Sigma_{t=u}^j b_t\eff_l(\beta_t) \geq \\
            \eff_l(\alpha_1\beta_1^{a_1}\alpha_1'\alpha_2\beta_2^{a_2}\alpha_2'\ldots\alpha_j\beta_j^{a_j}\alpha_j')+b_u - \Sigma_{t=u+1}^jb_tm|\beta_t| \geq \eff_l(\alpha_1\beta_1^{a_1}\alpha_1'\alpha_2\beta_2^{a_2}\alpha_2'\ldots\alpha_j\beta_j^{a_j}\alpha_j')
        \end{multline*}
        \myparagraph{Point \ref{phi:accepting:enum:2}}
        If $\eff_l(\beta_n) \geq 0$ than because of point \ref{phi:accepting:enum:1} we have:
        $$\eff_l(\alpha_1\beta_1^{a_1+b_1}\alpha_1'\alpha_2\beta_2^{a_2+b_2}\alpha_2'\ldots\alpha_{n}\beta_{n}^{a_{n}+b_n+i|\gamma|}) \geq \eff_l(\alpha_1\beta_1^{a_1}\alpha_1'\alpha_2\beta_2^{a_2}\alpha_2'\ldots\alpha_{n}\beta_{n}^{a_{n}})$$
        Hence, we can assume $\eff_l(\beta_n) < 0$. Therefore there has to be $j < n$ such that $\eff_l(\beta_j) > 0$. Let $j$ be the maximal such. 
        Observe, that:
        \begin{multline*}
            \eff_l(\alpha_1\beta_1^{a_1+b_1}\alpha_1'\alpha_2\beta_2^{a_2+b_2}\alpha_2'\ldots\alpha_{j}\beta_{j}^{a_{j}}\alpha_j'\beta_{j+1}^{a_{j+1}}\alpha_{j+1}'\ldots\alpha_n\beta_n^{a_n}) \geq \\
            \eff_l(\alpha_1\beta_1^{a_1}\alpha_1'\alpha_2\beta_2^{a_2}\alpha_2'\ldots\alpha_n\beta_n^{a_n})
        \end{multline*}
        $$b_j \eff_l(\beta_j) \geq b_j$$
        And for each $s \in [j+1, n]$
        $$b_s\eff_l(\beta_s) \geq b_s m |\beta_s|$$
        Hence:
        \begin{multline*}
            \eff_l(\alpha_1\beta_1^{a_1+b_1}\alpha_1'\alpha_2\beta_2^{a_2+b_2}\alpha_2'\ldots\alpha_{n}\beta_{n}^{a_{n}+b_n+i|\gamma|}) =\\
            \eff_l(\alpha_1\beta_1^{a_1+b_1}\alpha_1'\alpha_2\beta_2^{a_2+b_2}\alpha_2'\ldots\alpha_{j}\beta_{j}^{a_{j}}\alpha_j'\beta_{j+1}^{a_{j+1}}\alpha_{j+1}'\ldots\alpha_n\beta_n^{a_n}) + \Sigma_{s=j}^n b_s\eff_l(\beta_s) + i|\gamma|\eff_l(\beta_n) \geq \\
            \eff_l(\alpha_1\beta_1^{a_1}\alpha_1'\alpha_2\beta_2^{a_2}\alpha_2'\ldots\alpha_n\beta_n^{a_n}) + b_j - \Sigma_{s=j+1}^{n}b_sm|\beta_s| - mi|\lambda||\beta_n| \geq \eff_l(\alpha_1\beta_1^{a_1}\alpha_1'\alpha_2\beta_2^{a_2}\alpha_2'\ldots\alpha_n\beta_n^{a_n}) 
        \end{multline*}
        \myparagraph{Point \ref{phi:accepting:enum:3}}
        If $\eff_l(\gamma) \geq 0$ than the inequality follows from Point \ref{phi:accepting:enum:2}. Hence now we assume $\eff_l(\gamma) < 0$. Therefore there has to be $j \in [n]$ such that $\eff_l(\beta_j) > 0$. Let $j$ be the maximal such. Hence:
        Observe, that:
        \begin{multline*}
            \eff_l(\alpha_1\beta_1^{a_1+b_1}\alpha_1'\alpha_2\beta_2^{a_2+b_2}\alpha_2'\ldots\alpha_{j}\beta_{j}^{a_{j}}\alpha_{j}'\ldots\alpha_n\beta_{n}^{a_n}\lambda\gamma) \geq \\
            \eff_l(\alpha_1\beta_1^{a_1}\alpha_1'\alpha_2\beta_2^{a_2}\alpha_2'\ldots\alpha_{j}\beta_{j}^{a_{j}}\alpha_{j}'\ldots\alpha_n\beta_{n}^{a_n}\lambda\gamma)
        \end{multline*}
        $$b_j \eff_l(\beta_j) \geq b_j$$
        And for $s \in [j+1, n]$ we have:
        $$b_s \eff_l(\beta_s) \geq b_s m |\beta_s|$$
        \begin{multline*}
          \eff_l(\alpha_1\beta_1^{a_1+b_1}\alpha_1'\alpha_2\beta_2^{a_2+b_2}\alpha_2'\ldots\alpha_{n}\beta_{n}^{a_{n}+b_n+i|\gamma|}\lambda\gamma^{1+(k+1-i)|\beta_n|}) = \\
          \eff_l(\alpha_1\beta_1^{a_1+b_1}\alpha_1'\alpha_2\beta_2^{a_2+b_2}\alpha_2'\ldots\alpha_{j}\beta_{j}^{a_{j}}\alpha_{j}'\ldots\alpha_n\beta_{n}^{a_n}\lambda\gamma) + b_j\eff_l(\beta_j) + \\
          \Sigma_{s=j+1}^n b_s\eff_l(\beta_s) + i|\gamma|\eff_l(\beta_n) + (k+1-i)|\beta_n|\eff_l(\gamma) \geq \\
          \eff_l(\alpha_1\beta_1^{a_1}\alpha_1'\alpha_2\beta_2^{a_2}\alpha_2'\ldots\alpha_{j}\beta_{j}^{a_{j}}\alpha_{j}'\ldots\alpha_n\beta_{n}^{a_n}\lambda\gamma) + b_j - \Sigma_{s=j+1}^nb_sm|\beta_s| - (k+1)m|\beta_n||\gamma| = \\
        \eff_l(\alpha_1\beta_1^{a_1}\alpha_1'\alpha_2\beta_2^{a_2}\alpha_2'\ldots\alpha_{j}\beta_{j}^{a_{j}}\alpha_{j}'\ldots\alpha_n\beta_{n}^{a_n}\lambda\gamma)
        \end{multline*}
    \end{proof}
    \begin{claim}\label{clm:phi:cond:1:word}
        For each $i \in [k+1]$ run $\phi_i$ reads the same word.
        %For each $i, j \in [k+1]$ runs $\phi_i$ and $\phi_j$ read the same word.
    \end{claim}
    \begin{proof}
        Let us fix $i,j \in [k+1]$ such that $i < j$. 
        Observe, that $w(\phi_i) = w(\phi_j)$ if and only if 
        $w(\beta_{n}^{(j-i)|\gamma|}\lambda)$ and $w(\lambda\gamma^{(j-i)|\beta_n|})$.
        We have, that $w(\beta_n), w(\alpha_n') \in L(a^*)$.
        Hence also $\gamma, \lambda \in L(a^*)$, as they are parts of $\alpha_n'$.
        Therefore these two runs read the same word because 
        $$|\beta_{n}^{(j-i)|\gamma|}\lambda| = |\lambda\gamma^{(j-i)|\beta_n|}|$$
    \end{proof}
    \begin{claim}\label{clm:phi:cond:1:different}
        For each $i, j \in [k+1]$ such that $i \neq j$ we have that $\phi_i \neq \phi_j$.
    \end{claim}
    \begin{proof}
    \Wlog assume $i < j$. Observe, that it is enough to show, that 
    $$\beta_{n}^{(j-i)|\gamma|}\lambda \neq \lambda\gamma^{(j-i)|\beta_n|}$$
    Assume, for the sake of contradiction that:
     $$\beta_{n}^{(j-i)|\gamma|}\lambda = \lambda\gamma^{(j-i)|\beta_n|}$$
    Observe, that $\beta_n$ is not a prefix of $\lambda\gamma$, because then $\beta_n$ would be a prefix of $\alpha_n'$, which is not possible because of properties of $\pi_2$ (recall, that $\pi_2 = \alpha_n\beta_n^{a_n}\alpha_n'$ and $\beta_n$ is not a prefix of $\alpha_n'$). 
    Hence $\lambda\gamma$ is a prefix of $\beta_n$. 
    
     Hence for all $l \in [d]$ we have:
     $$\eff_l(\beta_{n}^{(j-i)|\gamma|}\lambda) = \eff_l(\lambda\gamma^{(j-i)|\beta_n|})$$
     Hence $\support(\beta_n) = \support(\gamma)$. Because $\lambda\gamma$ is a prefix of $\beta_n$ and $\lambda\gamma \neq \beta_n$ this contradicts the minimality of $|\alpha_n\beta_n|$.
    \end{proof}
    
Since we have proven Claims~\ref{clm:phi:cond:1:accepting},~\ref{clm:phi:cond:1:word}~and~\ref{clm:phi:cond:1:different}
we know that runs $\phi_i$ are all accepting, different and read the same word. This concludes the proof of Lemma~\ref{lem:decomp}.
\end{toappendix}

% Now we are ready to prove Lemma \ref{lem:second_loop}.
% \begin{proof}[Proof of Lemma \ref{lem:second_loop}]
%      Let constants $N_1, K_1, N_2, K_2, \ldots, N_n, K_n$ be the constants from Lemma \ref{lem:decomp} for $L, V$ and $n$. Let us set $C = \max_{i \in [n]} N_i + K_i + 1$ and take any $\rho$ satisfying conditions of Lemma \ref{lem:second_loop}. Observe, that $\rho$ satisfies also conditions of Lemma \ref{lem:decomp}. Hence it can be decomposed as $\rho=\alpha_1\beta_1^{a_1}\alpha_1'\alpha_2\beta_2^{a_2}\alpha_2'\ldots\alpha_{n+1}\beta_{n}^{a_{n}}\alpha_{n}'$. Recall, that for each $i \in [n]$ we have $N_i \leq C$ and observe that conditions 1-3, 5 and 6 are satisfied. To observe that condition 4 is satisfied it is enough to prove the following claim:
%     \begin{claim}
%         For each $i \in [n]$ we have $\beta_i \neq \varepsilon$.
%     \end{claim}
%     \begin{proof}
%         Assume, for the sake of contradiction, that for some $i$ we have $\beta_i = \varepsilon$. Therefore i-th block of letter a has at most length equal to:
%         $$|\alpha_i\beta_i^{a_i}\alpha_i'| = |\alpha_i\alpha_i'| \leq N_i + K_i < C$$
%         which contradicts the fact that $m_i \geq C$.
%     \end{proof}
% \end{proof}
%Now we proceed with the proof of Lemma \ref{lem:first_tool}.

%
\begin{toappendix}
% We are ready to prove the first tool of Section~\ref{sec:techniques}, namely Lemma~\ref{lem:first_tool}.

\subsection{Proof of Lemma \ref{lem:first_tool}}
   The proof is by contradiction. We assume that $L_1$ is recognised by a \kamb VASS $V$. Firstly we define a few constants. Then we define word $w_1, w_2, \ldots, w_{k+1} \in L_1$ and respective accepting runs $\rho_1, \rho_2, \ldots, \rho_{k+1}$. Then we decompose each of these runs using Claim \ref{clm:decomp} and apply pumping to each of the runs.
   Finally we get $k+1$ different runs $\rho_1', \rho_2', \ldots, \rho_{k+1}'$ over the same word, which contradicts the fact that $V$ is \kamb.
   
   Assume, for the sake of contradiction, that $L_1$ is recognised by a \kamb VASS $V$ and let $m$ be the maximal norm of a transition in $V$ (\mywlog $m \geq 2$).
   Let also $n$ be the maximal norm of an initial configuration. Let us fix a few constants. Let $N_0$ be equal to $2C+1$ (where $C$ is the constant given by Lemma \ref{lem:decomp}). Moreover, for $i \in [k+2]$ we define $N_i$ as $\max(N_{i-1} + 2, M(V, [d], M_i))$ where $M_i = \Sigma_{j=1}^{i-1} m(N_j! + 1) + n$ and $M(V, [d], M_i)$ is the constant given by Lemma \ref{lem:loop} for VASS $V$, subset of counters $[d]$, and initial value of the counters $M_i$. Observe, that because of this for each $i \in [k+1]$ we have $N_{i+1}! \geq (N_i+2)! > (N_i+1)!$ and for each $i \in [k+2]$ we have $N_i \geq 3$ (because $N_0 \in \N_+$). Let also $u$ be a word from $L$ such that $f(u) \geq 2N_{k+2}!$. Such word exists because $\sup{f} = \omega$.

   Having these constants and the word $u$ we define words $w_1, w_2, \ldots w_{k+1}$ as:
   $$w_i = a^{N_1!}ba^{N_2!}b\ldots a^{N_i!}ba^{N_0!}ba^{N_{i+2}!}ba^{N_{i+3}!}b\ldots a^{N_{k+2}!}bu$$
   Observe, that $w_i \in L$ because $N_i! \geq N_0!$ and let $\rho_i$ be an accepting run over $w_i$. 
   The intuition is, that $\rho_i$ has a nonnegative loop on all the counters in each part reading a block of letters a possibly except the one reading $i+1$ block. This is formalised in the following claim:
   \begin{claim}
       Let $j \in [k+2] \setminus \{i+1\}$ and let $\pi$ be the part of $\rho_i$ reading $j$-th block of letters $a$. Then $\pi = \alpha\beta\alpha'$ where $\beta$ is a nonnegative loop on all the counters and $|\beta| \leq N_j$. 
   \end{claim}
   \begin{proof}
   Observe, that for each $l \in [k+2]$ we have $N_0! \leq N_l!$. Recall, that $n$ is the maximal norm of an initial configuration and $m$ is the maximal norm of a transition. Observe also that before reading $j$-th block of letter $a$ each counter is bounded by:
   $$\Sigma_{l=1}^{j-1}m(N_l!+1) + n = M_{j}$$
   Hence by definition of $N_j$, Lemma \ref{lem:loop} and the fact that $|\pi| = N_j! \geq N_j$ we know, that we have a nonnegative loop $\beta$ in $\pi$ such that there exists $\alpha$, $\alpha'$ such that $\pi = \alpha\beta\alpha'$ and $|\beta| \leq N_j$.
   \end{proof}
   Now we have a nonnegative loop in each block possibly except the $i+1$ block. We want to find a loop in this block, which we will be able to pump, probably with the help of previous nonnegative loops in the earlier blocks. For this, we need Lemma \ref{lem:decomp}.
   We formalise our goal in the following claim, which is similar to Lemma \ref{lem:decomp}, but more specific to our situation. In particular conditions \ref{cond:clm_decomp_first}, \ref{cond:clm_decomp_second}, \ref{cond:clm_decomp_third} and \ref{cond:clm_decomp_fourth}-\ref{cond:clm_decomp_fifth}
   from Claim \ref{clm:decomp} correspond respectively to conditions \ref{cond:two:lem:decomp}, \ref{cond:three:lem:decomp}, \ref{cond:four:lem:decomp} and \ref{cond:five:lem:decomp} from Lemma \ref{lem:decomp}.
   \begin{claim}\label{clm:decomp}
       For each $i \in [k+1]$ run $\rho_i$ can be decomposed as 
       $$\rho_i = \alpha_1^i(\beta_1^i)^{a_1^i}\gamma_1^i\alpha_2^i(\beta_2^i)^{a_2^i}\gamma_2^i\ldots\alpha_{k+2}^i(\beta_{k+2}^i)^{a_{k+2}^i}\gamma_{k+2}^i\pi_i$$
       for some $a_1^i, a_2^i, \ldots, a_{k+2}^i \in \N_+$ and moreover we have:
       \begin{enumerate}
           \item For each $j < k+2$ we have that $w(\gamma_j^i) \in L(a^*b)$ and $w(\gamma_{k+2}^i) \in L(a^*)$\label{cond:clm_decomp_first}
           \item For each $j \in [k+2]$ we have that $w(\alpha_j^i), w(\beta_j^i) \in L(a^*)$.\label{cond:clm_decomp_second}            
           \item For each $j \in [k+2]$ we have that $\beta_j^i$ is a loop and $\beta_j^i \neq \varepsilon$\label{cond:clm_decomp_third}
           \item For each $j \in [k+2] \setminus \{i+1\}$ we have that $\beta_j^i$ is a nonnegative loop and $|\beta_j^i| \leq N_j$\label{cond:clm_decomp_fourth}
           \item We have that $\beta_{i+1}^i$ is a nonnegative loop on counters from $[d] \setminus \bigcup_{1 \leq j \leq i} \support(\beta_j^i)$, $|\beta_{i+1}^i| \leq N_0$ and $\beta_{i+1}^i$ is not a nonnegative loop on all the counters\label{cond:clm_decomp_fifth}
       \end{enumerate}
   \end{claim}
   \begin{proof}
       Let us fix $i \in [k+1]$. Let $\pi_i$  be part of $\rho_i$ reading word $bu$. Then $\rho_i=\pi\pi_i$ for some run $\pi$. Observe, that $\pi$ satisfies conditions of Lemma \ref{lem:decomp}. Hence:
       $$\pi = \alpha_1\beta_1^{a_1}\gamma_1\alpha_2\beta_2^{a_2}\alpha_2'\ldots \alpha_1\beta_{k+2}^{a_{k+2}}\alpha_{k+2}'$$
       Let us for each $j \in [k+2]$ set $\gamma_j=\alpha_j'$. Therefore:
       $$\rho_i = \alpha_1(\beta_1)^{a_1}\gamma_1\alpha_2(\beta_2)^{a_2}\gamma_2\ldots\alpha_{k+2}(\beta_{k+2})^{a_{k+2}}\gamma_{k+2}\pi_i$$
       Observe, that for simplicity we drop upper index $i$.
       Clearly conditions \ref{cond:clm_decomp_first} and \ref{cond:clm_decomp_second} are satisfied.
       Now we show, that the other conditions are also satisfied. We start with conditions \ref{cond:clm_decomp_third} and \ref{cond:clm_decomp_fifth}, which are significantly simpler to prove than Condition \ref{cond:clm_decomp_fourth}.
       \myparagraph{Condition \ref{cond:clm_decomp_third}}
       Because of Lemma \ref{lem:decomp} we have to only show that $\beta_j^i \neq \varepsilon$. Observe, that the number of letter a in the $j$-th block equals at least $N_0! \geq N_0 = C$ (where $C$ is constant from Lemma \ref{lem:decomp}). Hence $\beta_j^i \neq \varepsilon$. 
       %\lotodo{Może Zmienić na jedną stałą}
       \myparagraph{Condition \ref{cond:clm_decomp_fifth}}
       Because of Lemma \ref{lem:decomp} we have that $\beta_{i+1}$ is a nonnegative loop on counters from $[d] \setminus \bigcup_{1 \leq j \leq i} \support(\beta_j)$. Moreover, we also have:
       $$|\beta_{i+1}| \leq C \leq N_0$$
       The only thing left is that $\beta_{i+1}$ is not a nonnegative loop on all the counters.
       The idea is, that if $\beta_{i+1}$ is nonnegative on all the counters we can create an accepting run $\rho_i'$, which accepts word $v \notin L_1$, which is a contradiction.
       Therefore assume, towards contradiction, that $\beta_{i+1}^i$ is nonnegative on all the counters. Then run:
$$\alpha_1(\beta_1)^{a_1}\gamma_1\alpha_2(\beta_2)^{a_2}\gamma_2\ldots\alpha_{i+1}(\beta_{i+1})^{a_{i+1}+b}\gamma_{i+1}\ldots \alpha_{k+2}(\beta_{k+2})^{a_{k+2}}\gamma_{k+2}\pi_i$$
       where $b = \frac{2N_{i+1}!-N_0!}{|\beta_{i+1}^i|}$ is also an accepting run. Observe, that $b$ is well-defined because $|\beta_{i+1}^i| \leq N_0 \leq N_{i+1}$.
       This run reads the following word $v$:
       $$v = a^{m_1}ba^{m_2}b\ldots a^{m_{k+2}}bu$$
       where for $j \in [k+2] \setminus \{i+1\}$ we have $m_j = N_j!$ and $m_{i+1} = 2N_i!$. Observe that for $j \in [i-1] \cup [i+2, k+1]$ we have $N_{j}! < N_{j+1}!$. Moreover, we have $N_i! < 2N_{i+1}!$ and $2N_{i+1}! \leq (N_{i+1} + 1)! < N_{i+1}!$. Finally $N_{k+2}! < 2N_{k+2}! \leq f(u)$. Hence $v \notin L_1$ and therefore the run reading $v$ can not be an accepting run, which is a contradiction and therefore $\beta_{i+1}^i$ can not be nonnegative on all the counters. 
       \myparagraph{Condition \ref{cond:clm_decomp_fourth}}
       Observe, that for each $j \in [k+2] \setminus \{i+1\}$ by Lemma \ref{lem:decomp}:
       $$|\beta_j| \leq C \leq N_0 \leq N_j$$
       Hence, we have to prove that for each $j \in [k+2] \setminus \{i+1\}$ it occurs that $\beta_j$ is a nonnegative loop. Assume, towards contradiction, that there exists $j \in [k+2] \setminus \{i+1\}$ such that $\beta_j$ is not a nonnegative loop and let $j$ be the minimal such. We aim to argue that in this case, we would have at least $k+1$ different runs over the same word. 
       The idea is, that in the $j$th block of letters $a$ we also have another loop, which is nonnegative on all the counters and therefore if $\beta_j$ is not a nonnegative loop we have too many degrees of freedom and we are able to construct $k+1$ different runs over the same word.
       For constructing these runs, we need to be able to apply pumping to the run $\rho_i$. 
       For shortcut, for all $l \in \N$ and $i \in [k+2]$ we define $\psi_i^l = \alpha_j(\beta_j)^l\gamma_j$. We present the core of this pumping technique in the following claim:
       \begin{claim}\label{clm:pumping}
           For each $b_j \in \N$ there exist $ b \in \N$ such that for each $l \in [d]$ and $b_1, b_2, \ldots, b_{j-1} \geq b$
           \begin{multline*}
           \eff_l(\psi_1^{a_1+b_1}\psi_2^{a_2+b_2}\ldots\psi_{j-1}^{a_{j-1}+b_{j-1}}\alpha_j(\beta_j)^{b_j}) \geq  \eff_l(\psi_1^{a_1}\psi_2^{a_2}\ldots\psi_{j-1}^{a_{j-1}}\alpha_j)
           \end{multline*}
       \end{claim}
       \begin{proof}
       Recall that $m$ is the maximal norm of a transition and set $b=b_jm|\beta_j^i|$ and let us fix $l \in [d]$.
       If $\eff_l(\beta_j) \geq 0$ the inequality from the Claim \ref{clm:decomp} holds.
       Hence we can assume $\eff_l(\beta_j) < 0$. Because $\beta_j$ is a nonnegative loop on the counters from $[d] \setminus \bigcup_{1 \leq m < j} \support(\beta_m^i)$ there exists $n < j$ such that $\eff_l(\beta_n) > 0$. Hence
       \begin{multline*}
           \eff_l(\psi_1^{a_1+b_1}\psi_2^{a_2+b_2}\ldots\psi_{j-1}^{a_{j-1}+b_{j-1}}\alpha_j(\beta_j)^{b_j}) \geq \\ \eff_l(\psi_1^{a_1}\psi_2^{a_2}\ldots\psi_{j-1}^{a_{j-1}}\alpha_j) + b\eff_l(\beta_n^i) + b_j\eff_l(\beta_j^i) \geq \\
           \eff_l(\psi_1^{a_1}\psi_2^{a_2}\ldots\psi_{j-1}^{a_{j-1}}\alpha_j) + b_jm|\beta_j^i| - b_jm|\beta_j^i| =
           \eff_l(\psi_1^{a_1}\psi_2^{a_2}\ldots\psi_{j-1}^{a_{j-1}}\alpha_j)
           \end{multline*}
       \end{proof}
       First, we observe, the following fact about $\psi_j^{a_j}$, which will be useful for defining $k+1$ different runs over the same word:
       \begin{claim}\label{clm:nonegative_loop}
           There exist runs $\lambda, \delta, \lambda'$ such that $w(\lambda\delta) \in L(a^*)$, $\delta$ is a nonnegative loop and $\psi_j^{a_j} = \lambda\delta^m\lambda'$ for $m \in \N_+$.
       \end{claim}
       \begin{proof}
           Observe, that for each $l \in [d]$ we have:
           \begin{multline*}
           \eff_l(\psi_1^{a_1}\psi_2^{a_2}\ldots\psi_{j-1}^{a_{j-1}}) \leq 
           m|\psi_1^{a_1}\psi_2^{a_2}\ldots\psi_{j-1}^{a_{j-1}}|  \leq \Sigma_{m=1}^{j-1}m(N_m!+1) \end{multline*}
           Moreover $$|\psi_j^{a_j}| \geq N_j! > N_j$$
           Moreover
           $$\psi_1^{a_1}\psi_2^{a_2}\ldots\psi_{j-1}^{a_{j-1}}$$
           is a run. Hence because of the definition of $N_j$ there exist runs $\lambda$, $\delta$, $\lambda'$ such that $\lambda' \neq \varepsilon$, $\delta$ is a nonnegative loop and there exist $m \in \N+$ such that $\psi_j^{a_i} = \lambda\delta^m\lambda'$. Hence and because $\lambda' \neq \varepsilon$ we have that $w(\lambda\gamma) \in L(a^*)$.
       \end{proof} 
       Because $\delta$ is a nonnegative loop run
       $$\psi_1^{a_1}\psi_2^{a_2}\ldots\psi_{j-1}^{a_{j-1}}\lambda\delta^{m+1}\lambda'\psi_{j+1}^{a_{j+1}}\ldots\psi_{k+2}^{a_{k+2}}\pi_i$$
       is an accepting run. Because of Condition \ref{cond:seven:lem:decomp} from Lemma \ref{lem:decomp} we know that $\alpha_j\beta_j$ is a prefix of $\lambda\delta$ (hence $\lambda\delta^m = \alpha_j\beta_j\pi$ for some run $\pi$).
       Therefore this run can be written as:
       $$\psi_1^{a_1}\psi_2^{a_2}\ldots\psi_{j-1}^{a_{j-1}}\alpha_j\beta_j\pi\delta\lambda'\psi_{j+1}^{a_{j+1}}\ldots\psi_{k+2}^{a_{k+2}}\pi_i$$
       %$$\alpha_1^i(\beta_1^i)^{a_1^i}\gamma_1^i\alpha_2^i(\beta_2^i)^{a_2^i}\gamma_2^i\ldots\alpha_j^i\beta_j^i\pi\delta\lambda'\ldots\alpha_{k+2}^i(\beta_{k+2}^i)^{a_{k+2}^i}\gamma_{k+2}^i\pi_i$$
       Now we define runs $\phi_n$ for $n \in [k+1]$.
       $$\phi_n = \psi_1^{a_1+b}\psi_2^{a_2+b}\ldots\psi_{j-1}^{a_{j-1}+b}\alpha_j(\beta_j)^{1+|\delta|n}\pi(\delta)^{1+(k+1-n)|\beta_j|}\lambda'\psi_{j+1}^{a_{j+1}}\ldots\psi_{k+2}^{a_{k+2}}\pi_i$$
       %$$\phi_n = \psi_1^{a_1}\alpha_2^i(\beta_2^i)^{a_2^i+b}\gamma_2^i\ldots\alpha_j^i(\beta_j^i)^{1+|\delta|n}\pi\delta^{1+(k+1-n)|\beta_j^i|}\lambda'\ldots\alpha_{k+2}^i(\beta_{k+2}^i)^{a_{k+2}^i}\gamma_{k+2}^i\pi_i$$
       where $b$ is the maximal $b$ got from application of Claim \ref{clm:pumping} to $b_j = |\delta|, 2|\delta|, \ldots, (k+1)|\delta|$.
       Now it is enough to prove three claims \ref{clm:phi_accepting}, \ref{clm:phi_same_word} and \ref{clm:phi_different} to reach a contradiction with VASS $V$ being $k$-ambiguous.
       \begin{claim}\label{clm:phi_accepting}
       For each $n \in [k+1]$ we have that $\phi_n$ is an accepting run.   
       \end{claim}
       \begin{proof}
       Observe, that because $\beta_1^i, \beta_2^i, \ldots, \beta_{j-1}^i$ and $\delta$ are nonnegative it is enough to prove for each $l \in [d]$:
       $$\eff_l(\psi_1^{a_1+b}\psi_2^{a_2+b}\ldots\psi_{j-1}^{a_{j-1}+b}\alpha_j(\beta_j)^{|\delta|n}) \geq \eff_l(\psi_1^{a_1}\psi_2^{a_2}\ldots\psi_{j-1}^{a_{j-1}}\alpha_j)$$
       We get this inequality directly from Claim \ref{clm:pumping}.
       \end{proof}
       \begin{claim}\label{clm:phi_same_word}
       For each $n,m  \in [k+1]$ such that $n \neq m$ we have that $w(\phi_n) = w(\phi_m)$.
       \end{claim}
       \begin{proof}
       \Wlog let $n < m$ and observe that it is enough to show, that runs $$w(\pi\delta^{(m-n)|\beta_j|})= w((\beta_j)^{|\delta|(m-n)}\pi)$$
       Because of Claim \ref{clm:decomp}, Claim \ref{clm:nonegative_loop} and the way how $\pi$ was set (we have $\alpha_j^i\beta_j^i\pi = \lambda\delta^m$) we have that $w(\pi), w(\delta), w(\beta_j) \in L(a^*)$. 
       Hence equality $|\pi\delta^{(m-n)|\beta_j|}| = |(\beta_j)^{|\delta|(m-n)}\pi|$ concludes the proof.
       \end{proof}
       \begin{claim}\label{clm:phi_different}
           For each $n,m \in [k+1]$ such that $i \neq j$ we have $\phi_n \neq \phi_m$.
       \end{claim}
       \begin{proof}
       \Wlog let $n < m$ and observe that it is enough to show
       $$\pi\delta^{(m-n)|\beta_j^i|} \neq (\beta_j^i)^{|\delta|(m-n)}\pi$$
       Recall, that we assumed, that $\beta_j$ is not a nonnegative loop. Hence there exists $l \in [d]$ such that $\eff_l(\beta_j) < 0$.
       Let us fix such $l$. It is enough to prove, that 
       $$\eff_l(\pi\delta^{(m-n)|\beta_j^i|}) \neq \eff_l((\beta_j^i)^{|\delta|(m-n)}\pi)$$
       Observe, that from the nonnegativity of $\delta$ we have:
       $$\eff_l(\pi\delta^{(m-n)|\beta_j^i|}) \geq \eff_l(\pi)$$
       Moreover, we have:
       $$\eff_l((\beta_j^i)^{|\delta|(m-n)}\pi) < \eff_l(\pi)$$
       Hence clearly 
       $$\eff_l(\pi\delta^{(m-n)|\beta_j^i|}) \neq \eff_l((\beta_j^i)^{|\delta|(m-n)}\pi)$$
       \end{proof}
       Hence we reached a contradiction with VASS $V$ being \kamb and therefore for each $j \in [k+2] \setminus \{i+1\}$ we have that $\beta_j^i$ is a nonnegative loop. 
   \end{proof}
   Having these decompositions our goal is to create $k+1$ different runs over word $w = a^{n_1}ba^{n_2}b\ldots a^{n_{k+2}}bu$ where for each $j \in [k+2]$ we have $n_{j} = 2m^{k+3-j}\Pi_{l=j}^{k+2}N_l!$.
   Therefore for $i \in [k+1]$ we define the following runs:
   $$\rho_i' = \alpha_1^i(\beta_1^i)^{a_1^i+b_1^i}\gamma_1^i\alpha_2^i(\beta_2^i)^{a_2^i+b_2^i}\gamma_2^i\ldots\alpha_{k+2}^i(\beta_{k+2}^i)^{a_{k+2}^i+b_{k+2}^i}\gamma_{k+2}^i\pi_i$$
   where
   $$b_j^i = \begin{cases}
   \frac{n_j - N_0!}{|\beta_j^i|},& j = i + 1\\
   \\
   \frac{n_j - N_j!}{|\beta_j^i|},              & \text{otherwise}
\end{cases}$$
   Observe, that $b_j^i \in \N_+$. This is because
   $$n_j \geq N_{k+2}! \geq N_0!$$
   Moreover, for $j = i + 1$ we have $|\beta_j^i| \leq N_{0}$ and 
   therefore $|\beta_j^i|$ divides $n_j - N_0!$. Similarly when $j \neq i + 1$ we have $|\beta_j^i| \leq N_j$ and therefore $|\beta_j^i|$ divides $n_j - N_j!$. 
   
   Now we have to prove three things to conclude, that we have a contradiction with the assumption that $V$ is a \kamb VASS.
   Firstly, in Claim \ref{clm:w}, we show, that each run reads the word $w$.
   Secondly, in Claim \ref{clm:phi_accepting}, we show, that each $\rho_i'$ is a valid and accepting run.
   Finally, in Claim \ref{clm:different}, we show, that there is no $i, j \in [k+1]$ such that $i \neq j$ and $\rho_i' = \rho_j'$.
   \begin{claim}\label{clm:w}
       For each $i \in [k+1]$ run $\rho_i'$ reads a word $w$ such that
       $$w = a^{n_1}ba^{n_2}b\ldots a^{n_{k+2}}bu$$
       where
       $$n_j = 2m^{k+3-j}\Pi_{l= j}^{k+2}N_l!$$
   \end{claim}
   \begin{proof}
       Let us fix $i \in [k+1]$. Observe, that because of Claim \ref{clm:decomp} for each $j \in [k+2]$ we have that $(\beta_j^i)^{b_j^i}$ reads word $v_j^i$ where:
       $$
       v_j^i = \begin{cases}
           
      a^{n_j - N_0!},& j = i + 1\\
   \\
   a^{n_j - N_j!},              & \text{otherwise}
       \end{cases}
   $$
       Because $\rho_i'$ differs from $\rho_i$ only by additional repetition of $\beta_1^i, \beta_2^i, \ldots, \beta_{k+2}^i$ we have that $\rho_i'$ reads word:
       $$a^{m_1}ba^{m_2}b\ldots a^{m_{k+2}}bu$$
       where
       $$m_j =  \begin{cases}
           |v_j^i| + N_0!,& j = i + 1\\
   \\
           |v_j^i| + N_j!,              & \text{otherwise}
       \end{cases} = n_j$$
   This concludes the proof.
   \end{proof}
   \begin{claim}\label{clm:accepting}
       For each $i \in [k+1]$ run $\rho_i'$ is a valid and accepting run of $V$.
   \end{claim}
   \begin{proof}
       Let us fix $i \in [k+1]$. Recall that by Claim \ref{clm:decomp} for each $j \in [k+2] \setminus \{i+1\}$ we have that $\beta_j^i$ is a nonnegative loop. To prove that $\rho_i'$ is an accepting run and that it is also valid, which means no counter drops below zero, it is enough to prove, that for each $l \in [d]$ we have:
       \begin{multline}\label{eq:valid_run}
       \eff_l(\alpha_1^i(\beta_1^i)^{a_1^i}\gamma_1^i\alpha_2^i(\beta_2^i)^{a_2^i}\gamma_2^i\ldots\alpha_{i+1}^i(\beta_{i+1}^i)^{a_{i+1}^i}) \leq \\
       \eff_l(\alpha_1^i(\beta_1^i)^{a_1^i+b_1^i}\gamma_1^i\alpha_2^i(\beta_2^i)^{a_2^i+b_2^i}\gamma_2^i\ldots\alpha_{i+1}^i(\beta_{i+1}^i)^{a_{i+1}^i+b_{i+1}^i})
       \end{multline}
       We have two cases. Let 
       $$S = [d] \setminus \bigcup_{1 \leq j \leq i} \support(\beta_j^i)$$
       \myparagraph{Case 1: $l \in S$}
       By Claim \ref{clm:decomp} we have that $\beta_{i+1}^i$ is nonnegative on counter $l$. Hence all $\beta_1^i, \beta_2^i, \ldots, \beta_{i+1}^i$ are nonnegative on this counter. From this inequality \ref{eq:valid_run} follows. \\
       \myparagraph{Case 2: $l \notin S$}
       Therefore we have $l \in \bigcup_{1 \leq j \leq i} \support(\beta_j^i)$. Hence, at least one of $\beta_1^i, \beta_2^i, \ldots, \beta_i^i$ is strictly positive on this counter, and the others are nonnegative. Because $|\eff_l(\beta_{i+1}^i)| \leq m|\beta_{i+1}^i|$ it is enough to prove:
       $$\min_{j \in [i]} b_j^i \geq m|\beta_{i+1}^i|b_{i+1}^i = mn_{i+1} - mN_0!$$
       Observe that:
       \begin{multline*}
           \min_{j \in [i]} b_j^i = \min_{j \in [i]} \frac{n_j - N_j!}{|\beta_j^i|} \geq
           \min_{j \in [i]} (N_j-1)!(mn_{j+1} - 1) \geq
           mn_{j+1} - 1 \geq mn_{i+1} - mN_0!
       \end{multline*}
       This concludes this case and the proof of the whole claim.
   \end{proof}
   \begin{claim}\label{clm:different}
       For each $i, j \in [k+1]$ such that $i \neq j$ we have $\rho_i' \neq \rho_j'$.
   \end{claim}
   \begin{proof}
       Without loss of generality assume that $i < j \leq k+1 $. Observe, that it is enough to prove, that there exists $l \in [d]$ such that 
       \begin{multline}\label{eq:different_runs}
       \eff_l(\alpha_{i+1}^i(\beta_{i+1}^i)^{a_{i+1}^i + b_{i+1}^i}\gamma_{i+1}^i) \neq \eff_l(\alpha_{i+1}^j(\beta_{i+1}^j)^{a_{i+1}^j+b_{i+1}^j}\gamma_{i+1}^i)
       \end{multline}
       Let us take $l \in [d]$ such that $\eff_l(\beta_{i+1}^i) < 0$. Such $l$ exists because of Claim \ref{clm:decomp}. Now we have:
       \begin{multline}\label{eq:different_runs_first}
       \eff_l(\alpha_{i+1}^i(\beta_{i+1}^i)^{a_{i+1}^i + b_{i+1}^i}\gamma_{i+1}^i) = \eff_l(\alpha_{i+1}^i(\beta_{i+1}^i)^{a_{i+1}^i}\gamma_{i+1}^i) + b_{i+1}^i\eff_l(\beta_{i+1}^i) \leq \\
       m|\alpha_{i+1}^i(\beta_{i+1}^i)^{a_{i+1}^i}\gamma_{i+1}^i| - b_{i+1}^i = m(N_0! + 1) - \frac{n_{i+1} - N_0!}{|\beta_{i+1}^i|} \leq \\
       m(N_0! + 1) - \frac{n_{i+1} - N_0!}{N_0} \leq m(N_0! + 1) - m\frac{N_{i+1}!}{N_0}n_{i+2} + (N_0-1)! \leq \\ m(N_0! + 1) - 2m^2N_{k+2}! + (N_0-1)! \leq (m+1)N_{k+2}!-2m^2N_{k+2}! = (1-m)(2m+1)N_{k+2}! <\\
       -mN_{k+2}! \leq -mN_{i+1}!
       \end{multline}
       %\lotodo{Czy tutaj więcej wytłumaczenia?}
       Moreover we have:\begin{multline}\label{eq:different_runs_second}
\eff_l(\alpha_{i+1}^j(\beta_{i+1}^j)^{a_{i+1}^j+b_{i+1}^j}\gamma_{i+1}^i) \geq \eff_l(\alpha_{i+1}^j\gamma_{i+1}^i) \geq -m|\alpha_{i+1}^j\gamma_{i+1}^i| \geq -m(N_{i+1}!)
       \end{multline}
       From inequalities \ref{eq:different_runs_first} and \ref{eq:different_runs_second} we get property \ref{eq:different_runs}, which concludes the proof.
   \end{proof}
   Therefore we reached a contradiction with the fact that $V$ is a \kamb VASS and hence $L_1$ is not recognised by a \kamb VASS, which finishes the proof of Lemma~\ref{lem:first_tool}.
\end{toappendix}

\subsection{Semilinear image}\label{subsec:image}
Now, we develop the second tool for showing that language is not recognised by a $k$-ambiguous VASS. 
Before formulating the tool we have to provide a few definitions.
For any language $L \subseteq \{a,b\}^*$ such that for each $w \in L$ we have $\#_b(w) = l$ for some $l \in \N$ we define
$$im(L) = \{(a_1, a_2, \ldots, a_{l+1}) \mid a_1, a_2, \ldots, a_{l+1} \in \N, a^{a_1}ba^{a_2}b\ldots ba^{a_{l+1}} \in L\}$$
Given a base vector $b \in \Z^d$ and a finite set of period vectors $P = \{p_1, \ldots, p_n\} \subseteq \Z^d$, the linear set $L(b, P)$ is defined as
$$L(b, P) = \{b + a_1p_1 + \ldots + a_np_n \mid a_i \in \N, 1 \leq i \leq n\}$$
A semi-linear set is a finite union of linear sets. 
Now we are ready to formulate the second tool.
\begin{lemma}\label{lem:image}
Let $L \subseteq \{a,b\}^*$ be a language satisfying:
\begin{itemize}
    \item $L$ is recognised by $k$-ambiguous VASS $V$.
    \item There exists $n \in \N$ such that for each $w \in L$ we have $\#_b(w) = n$.
\end{itemize}
Then $im(L)$ is a semilinear set.
\end{lemma}

%\begin{toappendix}
\begin{proof}
The proof is based on Lemma~\ref{lem:decomp} and the fact that set of solution of a system of linear Diophantine inequalities is semilinear.
Notice first, that $L$ satisfies conditions of Lemma~\ref{lem:decomp}. 
    Therefore we can apply Lemma \ref{lem:decomp} to $L$ and $n+1$ and decompose each accepting run in $V$ in the way presented in Lemma \ref{lem:decomp}. 
    Let us fix some $\alpha_1, \ldots, \alpha_{n+1}$, $\alpha_1', \ldots, \alpha_{n+1}'$ and $\beta_1, \ldots, \beta_{n+1}$, initial configuration $q(c)$ and accepting configuration $p(f)$. 
    We call a run $\rho$ accepting with respect to $q(c)$ and $p(f)$ if $\rho$ is an accepting run of the VASS starting in $q(c)$ and ending in configuration $p(c')$ such that $c' \geq f$. 
    Let $K$ be the following language:
    $$K = \{ w \in L \mid \text{there exist } a_1, a_2, \ldots, a_{n+1} \in \N \text{ such that } \alpha_1\beta_1^{a_1}\alpha_1'\alpha_2\beta_2^{a_2}\alpha_2'\ldots\alpha_{n+1}\beta_{n+1}^{a_{n+1}}\alpha_{n+1}'$$
    $$\text{   is an accepting run with respect to $q(c)$ and $p(f)$ and reads $w$}\}$$
    Observe, that $K$ depends on chosen $\alpha_i, \beta_i$ and $\alpha_i'$. Moreover, observe, that because of the constant given by Lemma $\ref{lem:decomp}$ we have only a finite number of possibilities of $\alpha_1, \ldots, \alpha_{k+1}$, $\alpha_1', \ldots, \alpha_{k+1}'$, $\beta_1, \ldots, \beta_{k+1}$, initial configuration $q(c)$ and accepting configuration $p(f)$. 
    Moreover, semilinear sets are closed under a finite union. Therefore to conclude, that $im(L)$ is a semilinear set it is enough to show that $im(K)$ is a semilinear set.  
    Notice, that from conditions of Lemma \ref{lem:decomp}, we know, that only $\alpha_i'$ (for $i \in [n]$) contain letter $b$, each exactly one letter at the last position. Hence:
    $$im(K) = \{(|\beta_1|a_1 + |\alpha_1| + |\alpha_1'| -1, |\beta_2|a_2 + |\alpha_2| + |\alpha_2'| -1, \ldots,$$$$, |\beta_n|a_n + |\alpha_n| + |\alpha_n'| -1, |\beta_{k+1}|a_{n+1} + |\alpha_{n+1}| + |\alpha_{n+1}'|) \mid $$$$\text{such that  $\alpha_1\beta_1^{a_1}\alpha_1'\alpha_2\beta_2^{a_2}\alpha_2'\ldots\alpha_{n+1}\beta_{n+1}^{a_{k+1}}\alpha_{n+1}'$ is an accepting run with respect to $q(c)$ and $p(f)$}\}$$
    Therefore it is enough to show, that:
    $$A = \{(a_1, a_2, \ldots, a_{n+1}) \mid \text{such that } $$$$\alpha_1\beta_1^{a_1}\alpha_1'\alpha_2\beta_2^{a_2}\alpha_2'\ldots\alpha_{n+1}\beta_{n+1}^{a_{n+1}}\alpha_{n+1}'$$$$ \text{is an accepting run with respect to $q(c)$ and $p(f)$}\}$$
    is a semilinear set. We have two cases. Either $A = \emptyset$, hence semilinear. This case occurs if for any $a_1, a_2, \ldots, a_{n+1} \geq 1$ we do not have an accepting run with respect to $q(c)$ and $p(f)$. Otherwise, we will show semilinearity, by providing a system of linear inequalities for $a_1, a_2, \ldots, a_{n+1}$. 
    It is enough because in \cite{semilinear} it was shown, that the set of solutions of a system of linear inequalities is a semilinear set. The goal of this system of linear inequalities is to express, that after each prefix of a run, we are non-negative on all of the counters. 
    Moreover, we want also to express the acceptance condition. Therefore for each counter $i$ we write the following inequalities:
    \begin{itemize}
        \item For each transition $t$ and each  $\alpha_j$ such that there exist $u$ and $v$ such that $\alpha_j=utv$: $$c_i + \eff_i(\alpha_1\beta_1^{a_1}\alpha_1'\alpha_2\beta_2^{a_2}\alpha_2'\ldots\alpha_{j-1}\beta_{j-1}^{a_{j-1}}\alpha_{j-1}') + \eff_i(ut) \geq 0 $$
        \item For each transition $t$ and each $\alpha_j'$ such that there exist $u$ and $v$ such that $\alpha_j'=utv$:
        $$c_i + \eff_i(\alpha_1\beta_1^{a_1}\alpha_1'\alpha_2\beta_2^{a_2}\alpha_2'\ldots\alpha_{j-1}\beta_{j-1}^{a_{j-1}}\alpha_{j-1}'\alpha_j\beta_j^{a_j})  + \eff_i(ut) \geq 0 $$
        \item For each transition $t$ and each $\beta_j$ such that there exist $u$ and $v$ such that $\beta_j=utv$:
            \begin{itemize}
                \item If $\eff_i(\beta_j) \leq 0$:
                $$c_i + \eff_i(\alpha_1\beta_1^{a_1}\alpha_1'\alpha_2\beta_2^{a_2}\alpha_2'\ldots\alpha_{j-1}\beta_{j-1}^{a_{j-1}}\alpha_{j-1}'\alpha_j\beta_j^{a_j-1}) + \eff_i(ut) \geq 0$$
                \item Otherwise:
                $$c_i + \eff_i(\alpha_1\beta_1^{a_1}\alpha_1'\alpha_2\beta_2^{a_2}\alpha_2'\ldots\alpha_{j-1}\beta_{j-1}^{a_{j-1}}\alpha_{j-1}'\alpha_j) + \eff_i(ut) \geq 0$$
            \end{itemize}
        \item Acceptance condition:
        $$c_i + \eff_i(\alpha_1\beta_1^{a_1}\alpha_1'\alpha_2\beta_2^{a_2}\alpha_2'\ldots\alpha_{k+1}\beta_{n+1}^{a_{n+1}}\alpha_{n+1}') \geq f_i$$
        \item Condition, that each $a_i$ is positive (this is needed because of conditions of Lemma \ref{lem:decomp}):
        $a_i \geq 1$
    \end{itemize}
    
    In other words, this system of inequalities ensures, that each transition in the sequence can be fired, check the acceptance condition and ensures that each $a_i$ is positive. We have shown, that the set $A$ is semilinear and hence $im(K)$ is a semilinear set. Therefore, because semilinar sets are closed under a finite union we have that $im(L)$ is a semilinear set.
\end{proof}
%\end{toappendix}

%The proof of Lemma~\ref{lem:image} is based on Lemma~\ref{lem:decomp} and the fact that set of solution of a system of linear Diophantine inequalities is semilinear.
%We can extend definition of $im(L)$ to different letters by setting:
%$$im_{c_1, c_2}(L) = \setof{(a_1, a_2, \ldots, a_{l+1})}{a_1, a_2, \ldots, a_{l+1} \in \N, c_1^{a_1}c_2c_1^{a_2}c_2\ldots c_2c_1^{a_{l+1}} \in L}$$
Then, as shown in Lemma \ref{lem:closure_prop}, $k$-ambiguous VASS languages are closed under intersection with regular languages. As mentioned below languages $K_n$ are regular for each $n \in \N$, the following corollary holds:

\begin{corollary}
Let $a, b \in \Sigma$, $L \subseteq \Sigma^*$ be a language recognised by a $k$-ambiguous VASS and for $n \in \N$ let $K_n \subseteq \{a,b\}^*$ be the language of words containing exactly $n$ letters $b$. Then for any $n \in N$ we have that $im(L \cap K_n)$ is a semilinear set.
\end{corollary}
% !TEX root = main.tex
\section{Properties}\label{sec:properties}
In this section we present several properties of languages of boundedly-ambiguous Vector Addition Systems with States.
\subsection{Closure properties}
First, we investigate the closure properties of boundedly-ambiguous languages.

%\wctodo{Może to podzielić na dwie podsekcje: własności domknięcia oraz inkluzje klas.}

\begin{lemma}\label{lem:closure_prop}
    If $L_1$ and $L_2$ are recognised by a $k_1$-ambiguous and $k_2$-ambiguous VASS respectively then:
    \begin{itemize}
        \item $L_1 \cap L_2$ is recognised by a $(k_1 \cdot k_2)$-ambiguous VASS;
        \item $L_1 \cup L_2$ is recognised by a $(k_1 + k_2)$-ambiguous VASS.
    \end{itemize}
    Moreover, class of languages of boundedly-ambiguous VASSs is not closed under complementation and commutative closure. 
\end{lemma}
%\lotodo{Zastanowić się nad: homomorphism, concatenation, kleene star}

\begin{proof}
    We split the proof into several parts, each corresponding to the closure properties under one operation.
    \myparagraph{Intersection} Let $L_1$ and $L_2$ be languages recognised respectively by $d_1$-dimensional $k_1$-ambiguous VASS $A_1$ and $d_2$-dimensional $k_2$-ambiguous VASS $A_2$. Language $L_1 \cap L_2$ can be recognised by the standard
synchronised product of $A_1$ and $A_2$. It is easy to observe that the product is a $(k_1 \cdot k_2)$-ambiguous $(d_1+d_2)$-VASS.
    \myparagraph{Union}
    Let $L_1$ and $L_2$ be languages recognised by $d_1$-dimensional $k_1$-ambiguous VASS $A_1$ and $d_2$-dimensional $k_2$-ambiguous VASS $A_2$. The idea is to recognise $L_1 \cup L_2$ by taking union VASS $A_1 \cup A_2$, which is clearly $k_1+k_2$-ambiguous and $\max(d_1, d_2)$-dimensional.
    
    \myparagraph{Complementation}
    This comes from a general fact, that coverability languages are not closed under complementation. Language $L = a^nb^{\leq n}$ is recognised even by a deterministic VASS (hence also by one with bounded-ambiguity). On the other hand, in \cite{CzerwinskiLMMKS18} it was shown, that every two disjoint coverability VASS\footnote{In fact this result was shown for a wider class of Well-structured transition system (WSTS)} languages $L_1$ and $L_2$ are regular separable, which means there exists a regular language $L_3$ such that $L_1 \subseteq L_3$ and $L_2 \cap L_3 = \emptyset$. Because $L$ is a coverability VASS language its complement can be a coverability VASS language if and only if $L$ is a regular language.
    Clearly $L_1$ is not and this can be shown using the Pumping Lemma for regular languages.
    
    \myparagraph{Commutative closure}
    Boundedly-ambiguous languages are not closed under commutative closure.
    Let us consider language $L = a^nb^{\leq n}$, which is recognised by a deterministic VASS. Its commutative closure is equal to $L_1 = \{w \mid \#_a(w) \geq \#_b(w)\}$.
    Assume, towards contradiction, that $L_1$ is recognised by a $k$-ambiguous VASS for $k \in \Nplus$.
    Because boundedly-ambiguous languages are closed under intersection with regular languages also language $L_1 \cap L(b^*a^*) = b^na^{\geq n}$ is recognised by a $k$-ambiguous VASS. Let $V$ be $k$-ambiguous VASS recognising $b^na^{\geq n}$. By Lemma \ref{lem:loop} we can take $N$ such, that while accepting $b^Na^N$ VASS $V$ will fire a non-negative loop on all of the counters. This loop reads $b^l$ for some $l \in \N$. Because VASS $V$ is $k$-ambiguous $l \geq 1$.  Hence we can fire this loop one more time and accept $b^{N+l}a^N$, which is not in the language $b^na^{\geq n}$. Therefore we reached a contradiction and boundedly-ambiguous languages are not closed under commutative closure.
    
    \end{proof}
Using the fact that regular languages are unambiguous (i.e. $1$-ambiguous) VASS languages and Lemma~\ref{lem:closure_prop} we can formulate the following remark about closure of boundedly-ambiguous languages under intersection with unambiguous (hence also regular) languages.
\begin{remark}
    For each $k \in \N_+$ the class of $k$-ambiguous VASS languages is closed under intersection with unambiguous VASS languages. Hence, the same holds for the intersection with regular languages.
\end{remark}
We complement Lemma \ref{lem:closure_prop} with Lemma \ref{lem:union} and Conjecture \ref{conj:intersection}.

\begin{lemma}\label{lem:union}
For each $k_1, k_2 \in \N_+$ there exist languages $L_1, L_2$, which are respectively recognised by a $k_1$ and $k_2$ ambiguous VASS such that language $L_1 \cup L_2$ is not recognised by an $n$-ambiguous VASS for $n \in [k_1+k_2 - 1]$.
\end{lemma}

\begin{proof}
    Let $k = k_1 + k_2$. For $i \in [k]$ let us define language $U_i = \{a^{n_1}ba^{n_2}ba^{n_3}b\ldots a^{n_{k+1}} \mid  n_i \geq n_{i+1} \}$. Observe, that $U_i$ is a language of an unambiguous VASS.
    Let $L_1 = \bigcup_{i=1}^{k_1} U_i$ and $L_2 = \bigcup_{i=k_1+1}^{k}$. By Lemma \ref{lem:closure_prop} languages $L_1$ and $L_2$ are respectively recognised by a $k_1$ and $k_2$ ambiguous VASS. 
    By applying Lemma \ref{lem:first_tool_simple} to $k-1$ we get that language $L_1 \cup L_2$ is not recognised by a $n$-ambiguous VASS for $n \in [k-1]$.
\end{proof}
\begin{conjecture}\label{conj:intersection}
For each $k_1, k_2 \in \N_+$ there exists languages $L_1, L_2$, which are respectively recognised by a $k_1$ and $k_2$ ambiguous VASS such that language $L_1 \cap L_2$ is not recognised by a $n$-ambiguous VASS for $n \in [k_1 \cdot k_2 - 1]$.

\end{conjecture}

\subsection{Expressiveness}\label{subsec:expressivity}
Firstly, we show that the hierarchy of $k$-ambiguous VASS is strict.
\begin{lemma}
    For every $k \in \N_+$ there exists language $L \in \xamb{(k+1)} \setminus \xamb{k}$. 
\end{lemma}
\begin{proof}
    Let $L = \{a^{n_1}ba^{n_2}ba^{n_3}b\ldots a^{n_{k+2}} \mid  \exists_{i \in [k+1]} \ n_i \geq n_{i+1} \}$.
    Because of Lemma \ref{lem:first_tool_simple} language $L$ is not recognised by a \kamb VASS.
    On the other hand it is a union of $k+1$ unambiguous VASS languages. For $i \in [k+1]$ let us define $L_i = \{a^{n_1}ba^{n_2}ba^{n_3}b\ldots a^{n_{k+2}} \mid  n_i \geq n_{i+1} \}$. 
    Observe, that $L = \bigcup_{i=1}^{k+1} L_i$. Hence, by Lemma \ref{lem:closure_prop}, $L$ is recognised by a $(k+1)$-ambiguous VASS.
\end{proof}
In terms of the classes of languages they define, boundedly-ambiguous VASS can express strictly more than deterministic ones. Moreover, they are strictly less expressive than non-deterministic ones.
We prove this in Lemmas \ref{lem:u-d} and \ref{lem:n-u}.
\begin{lemma}\label{lem:u-d}
    There exists language $L \in \xamb{1} \setminus \dett$.
\end{lemma}
We show an even stronger Lemma \ref{lem:u-h}. It is stronger because we have $\dett \subseteq \hist$.
\begin{lemma}\label{lem:u-h}
    There exists language $L \in \xamb{1} \setminus \hist$.
\end{lemma}
\begin{figure}
    \centering
    \begin{tikzpicture}
        \begin{scope}[initial text = \textcolor{red}{}]
            \node[state, initial] (A) at (0,0) {$q_1$};
            \node[state] (B) at (3,0) {$q_2$};
            \node[state, accepting] (C) at (6,0) {$q_3$};
            \node[state, accepting] (D) at (9,0) {$q_4$};
        \end{scope}
        
        \begin{scope}[>={Stealth[black]},
                      every edge/.style={draw=black,very thick}]
            \path [->] (A) edge [above] node {$b$} (B);
            \path [->] (B) edge [above] node {$a, +1$} (C);
            \path [->] (C) edge [above] node {$b, -1$} (D);
            \path [->] (A) edge [loop above] node {$a, b$} (A);
            \path [->] (C) edge [loop above] node {$a, +1$} (C);
            \path [->] (D) edge [loop above] node {$b, -1$} (D);
        \end{scope}
    \end{tikzpicture}
\caption{Unambiguous VASS recognising (starting from zero) $\{a,b\}^*ba^{n > 0}b^{\leq n}$}\label{fig:u-h}
\end{figure}
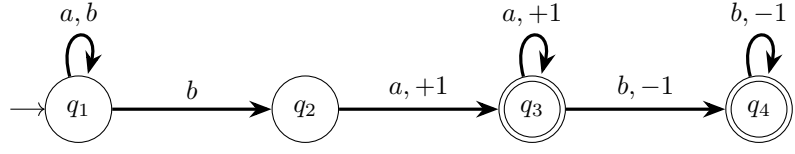
\begin{proof}
     Let $L = \{a,b\}^*ba^{n>0}b^{\leq n}$. Observe, that $L$ is recognized by an unambiguous VASS depicted in Figure \ref{fig:u-h}. The only point of nondeterminism is in state $q_1$ and there is only one way to guess when the last block, ending the word in the form $a^{n > 0}b^{\leq n}$ comes.
    
     On the other hand, assume, towards contradiction, that $L$ is recognised by a history-deterministic VASS. It is easy to see, that $L_1 = a^{n>0}b^{\leq n}$ is a language of a deterministic VASS. Hence it is also history-deterministic. In \cite{history-deterministic} it was shown, that history-deterministic VASS languages are closed under union. Hence $L_2 = L \cup L_1 = \{a,b\}^*a^{n > 0}b^{\leq n}$ is also history-deterministic VASS language. However, in~\cite{history-deterministic} it was also shown, that $L_2$ is not a history-deterministic VASS language. Hence we get, that $L$ is not a history-deterministic VASS language.

\end{proof}
\begin{lemma}\label{lem:n-u}
    There exists language $L \in \xndet{1} \setminus \amb$.
\end{lemma}
\begin{figure}
    \centering
    \begin{tikzpicture}
        \begin{scope}[initial text = \textcolor{red}{}]
            \node[state, initial] (A) at (0,0) {$q_1$};
            \node[state] (B) at (2,0) {$q_2$};
            \node[state] (C) at (4,0) {$q_3$};
            \node[state] (D) at (6,0) {$q_4$};
            \node[state, accepting] (E) at (8,0) {$q_5$};
        \end{scope}
        
        \begin{scope}[>={Stealth[black]},
                      every edge/.style={draw=black,very thick}]
            \path [->] (A) edge [above] node {$b$} (B);
            \path [->] (B) edge [loop above] node {$a, +1$} (B);
            \path [->] (B) edge [above] node {$b$} (C);
            \path [->] (A) edge [loop above] node {$a, b$} (A);
            \path [->] (C) edge [loop above] node {$a, -1$} (C);
            \path [->] (C) edge [above] node {$b$} (D);
            \path [->] (D) edge [loop above] node {$a,b$} (D);
            \path [->] (D) edge [above] node {$b$} (E);
        \end{scope}
    \end{tikzpicture}
\caption{VASS recognising (starting from zero) $\{a^{n}ba^{m}ba^{k} \mid n, m, k \in \N, (n \geq m \lor n \geq k) \}$}\label{fig:n-u}
\end{figure}
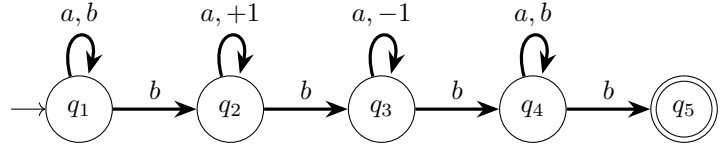
\begin{proof}
    Let $L = \setof{ua^nba^mbv}{u,v \in L((a^*b)^*), n \geq m}$. Let us fix $k \in \N_+$. We show, that $L$ is not recognised by a $k$-ambiguous VASS. Assume, towards contradiction, that it is.
    Hence language $L_k = L \cap L((a^*b)^{k+2})$ is also recognised by a $k$-ambiguous VASS.
    Let $f$ be a function such that $f : \set{\epsilon} \rightarrow \N_\omega$ and $f(\epsilon) = \omega$.
    Therefore, by Lemma \ref{lem:first_tool}, we get, that $L_k$ is not recognised by a $k$-ambiguous VASS, contradiction. Hence $L$ is not recognised by a $k$-ambiguous VASS.
    On the other hand, it can be recognised by a 1-dimensional VASS, which is presented in Figure \ref{fig:n-u}.
\end{proof}

Observe, that because each word is read by only finite number of accepting runs, one can see bounded-ambiguity as an extension of the notion of determinism.
A second extension of the determinism is history-determinism. 
In \cite{history-deterministic} it was shown, that history-deterministic VASSs can express more than deterministic ones and less than nondeterministic ones.
Up to now, there has been no comparison of the expressive power of history-deterministic and bounded-ambiguous VASSs. 
Now, we show, that these language classes are incomparable, which means that there exists a language recognised by an unambiguous VASS, which is not a language of history-deterministic VASS and language which is recognised by a history-deterministic VASS and it is not recognised by a $k$-ambiguous VASS for any $k \in \N_+$.
Recall, that in Lemma \ref{lem:u-h} we have shown that there exists $L \in \xamb{1} \setminus \hist$. Hence it is enough to show the following Lemma:
\begin{figure}
    \centering
    \begin{tikzpicture}
        \begin{scope}[initial text = \textcolor{red}{}]
            \node[state, initial] (A) at (0,0) {$q_1$};
            \node[state, accepting] (B) at (3,0) {$q_2$};
            \node[state, accepting] (C) at (6,0) {$q_3$};
        \end{scope}
        
        \begin{scope}[>={Stealth[black]},
                      every edge/.style={draw=black,very thick}]
            \path [->] (A) edge [above] node {$b, (0,1,0)$} (B);
            \path [->] (B) edge [loop above] node {$a, (0,-1,1)$} (B);
            \path [->] (B) edge [bend right] node [below]{$a, (0,0,0)$} (C);
            \path [->] (C) edge [bend right] node [above]{$a, (-1,0,0)$} (B);
            \path [->] (A) edge [loop above] node {$a, (1,0,0)$} (A);
            \path [->] (C) edge [loop above] node {$a, (0, 2, -1)$} (C);
        \end{scope}
    \end{tikzpicture}
\caption{VASS recognising (starting from zero) $\{a^{n}ba^{m} \mid n, m\in \N, (m \leq  2n + 2^{n+2}-1) \}$}\label{fig:h-a}
\end{figure}
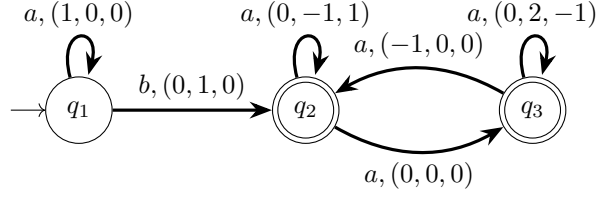
\begin{lemma}\label{lem:h-a}
    There exists language $L \in \hist \setminus \amb$.
\end{lemma}
\begin{proof}
    Let $L = a^nba^{2n+2^{n+2}-1}$. Observe, that $im(L)$ is not a semilinear set.
    Hence, due to Lemma \ref{lem:image}, for every $k \in \N_+$ we have that $L$ is not recognised by a $k$-ambiguous VASS.
    Hence $L \notin \amb$. On the other hand it is recognised by a history-deterministic VASS presented in the Figure \ref{fig:h-a}.
    It is history-deterministic, because the best option is to fire loops in states $q_2$ and $q_3$ as long as possible.
    In this way one can read words $a^nba^k$ for each $k \in \N$ such that
    $$k \leq 2 \cdot \Sigma_{i=1}^{n-1}(2^i+1) + 2^n + 1 + 2^n = 2 \cdot (2^n - 1) + 2 \cdot n + 2^{n+1} + 1 = 2 \cdot n + 2^{n+2} - 1$$
    %Language $L = \{a^nba^mba^k \mid n \geq m \lor n \geq k\}$ from Lemma \ref{lem:david} is not recognised by an unambiguous VASS. On the other hand, it is recognised by a history-deterministic 2-VASS, which is presented in Figure \ref{fig:h-u}.
\end{proof}   

% !TEX root = main.tex
\section{Proof of Theorem~\ref{thm:main}}\label{sec:undecidability}
In this section we prove Theorem \ref{thm:main}.
We start by introducing \emph{Lossy counter machines} (LCMs) \cite{lcm1, lcm2}.  Formally, an LCM is 
$M = \tupple{Q, Z, \Delta}$ where $Q = \{l_1, \ldots , l_m\}$ is a finite set of states, $Z = (z_1, \ldots , z_n)$ are $n$
counters, and $\Delta \subseteq Q \times OP(Z) \times Q$, where $OP(Z) = \{\inc, \dec, \ztest, \ski\}^n$.
A configuration of $M$ is $q(a)$ where $ q \in Q$ and $a = (a_1, . . . , a_n) \in \N^n$. There is a
transition $q(a) \trans{t} q'(b)$ if there exists $t \in \Delta$ such that $t = (q,op,q')$ and for each $i \in [1,n]$:
\begin{itemize}
    \item If $op_i = \inc$ then $b_i \leq a_i + 1$
    \item If $op_i = \dec$ then $b_i \leq a_i - 1$
    \item If $op_i = \ski$ then $b_i \leq a_i$
    \item If $op_i = \ztest$ then $a_i = b_i = 0$
\end{itemize}
Observe, that counters can nondeterministically decrease at each step. A run of M is a finite sequence $q_1(a_1) \trans{t_1} q_2(a_2) \trans{t_2} \ldots \trans{t_{n-1}} q_n(a_n)$. Given a configuration $q(u)$, the reachability set of $q(u)$ is the set of all configurations reachable from $q(u)$
via runs of M. We denote this set as $\reach(q(u))$. For simplicity we denote by $\reach(q)$ the set of configurations reachable from $q(\vec{0})$.
It was shown in \cite{lcm2} that the problem of deciding whether for a configuration $q(u)$ and LCM $M$ set $\reach(q(u))$
configuration is finite, is undecidable. Due to \cite{DBLP:conf/concur/AlmagorY22} even if $u$ is always equal to $\vec{0}$ the problem is still undecidable.
We call this problem \textbf{$0$-finite reach}. We prove Theorem \ref{thm:main} by reducing from $0$-finite reach.
The proof is similar to the proofs of undecidability of regularity \cite{regularity} and determinization \cite{DBLP:conf/concur/AlmagorY22}.
The rest of the section is devoted to the proof of Theorem \ref{thm:main}. 

    Firstly, we present an overview of the proof. For each LCM $M_1$ with an initial state we create another LCM $M$ and initial state $l_0$ with the same answer to the $0$-finite reach problem.
    Then, we define a language $L_{M, l_0}$, which intuitively encodes the correct runs of $M$. For technical reasons, namely because coverability VASSs are well-suited
    for recognising languages similar to $a^nb^{\leq n}$ we encode the correct runs in reverse, that means from the final configuration to the initial one. 
    Moreover, because we work with LCMs, it is better for us to work with the complement of $L_{M, l_0}$, that is $\widebar{L_{M, l_0}}$. 
    In such a way we get at the end a language for which Lemma \ref{lem:first_tool} is useful.
    Then, the proof of Theorem \ref{thm:main} is split into three claims. Claim \ref{clm:recognised} states, that
    $\widebar{L_{M,l_0}}$ is recognised by an effectively constructible $1$-dimensional VASS $A$. Claim \ref{clm:regular} provides, that if \reach($l_0$) is finite then $L_{M,l_0}$ is a regular language.
    Finally, Claim \ref{clm:not_k_amb} states, that if \reach($l_0$) is not finite then $L_{M,l_0}$ is not recognised by a VASS from the class of all boundedly-ambiguous VASS.
    All these claims give a direct reduction from $0$-finite reach to deciding whether a language of a $1$-VASS belongs to class $C$ of languages, which contains all regular languages and is contained in the class of all boundedly-ambiguous VASS languages. Thus they conclude the proof of Theorem \ref{thm:main}.  

    Let us fix LCM $M_1 = \tupple{Q_1, Z_1, \Delta_1}$ with $Q = \{l_1,l_2, \ldots l_m\}$ and $Z = \{z_1, z_2, \ldots, z_n\}$. Let $l_0$ be the initial state of $M_1$. We add to $M_1$ two states: $q_1$ and $q_2$ and for $i \in [1,m]$ transitions $t_i$ from $l_i$ to $q_1$ with no effect on the counters. 
    In addition, for each $i \in [2,n]$ we add a transition from $q_1$ to $q_1$ decrementing the $i$th counter and incrementing the first counter. We also add two transitions decrementing the first counter. 
    The first one goes from $q_1$ to $q_2$. 
    The second one goes from $q_2$ to $q_2$. 
    In such a way we obtain LCM $M = \tupple{Q, Z, \Delta}$. The sketch of the construction of $M$ is presented in the Figure \ref{fig:sketch}.
    \begin{figure}
        \centering
        \begin{tikzpicture}
            \begin{scope}[initial text = \textcolor{red}{}]
                \node[state] (A) at (0,6) {$l_1$};
                \node[state] (B) at (0,4) {$l_2$};
                \node[] (C) at (0,2) {$\vdots$};
                \node[state] (D) at (0,0) {$l_m$};
                \node[state] (E) at (3,0) {$q_1$};
                \node[state] (F) at (6,0) {$q_2$};
            \end{scope}
            
            \begin{scope}[>={Stealth[black]},
                          every edge/.style={draw=black,very thick}]
                \path [->] (A) edge [above] node {} (E);
                \path [->] (B) edge [above] node {} (E);
                \path [->] (C) edge [above] node {} (E);
                \path [->] (D) edge [above] node {} (E);
                \path [->] (E) edge [loop below] node {$\inc(c_1), \dec(c_i)$} (E);
                \path [->] (E) edge [above] node {$\dec(c_1)$} (F);
                \path [->] (F) edge [loop above] node {$\dec(c_1)$} (F);
            \end{scope}
        \end{tikzpicture}
    \caption{Sketch of the construction of LCM $M$.}\label{fig:sketch}
    \end{figure}
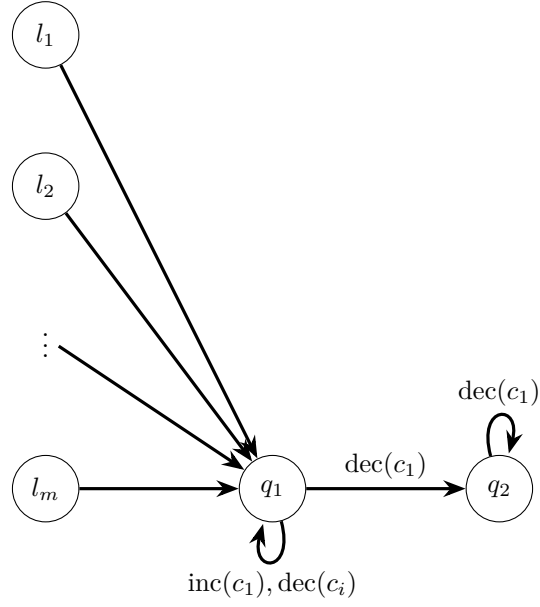
    Observe, that because each of the added transitions does not increase the sum of the counters, from $q_1$ we can go only to $q_2$ and later we can only stay in $q_2$ the answer for $0$-finite reach is the same for both: $M_1$ and $M$. Hence, we proceed later with $M$.
    We encode each configuration $q(a_1, a_2, \ldots, a_n)$ as a word over $\Sigma = Q \cup Z$ as $qz_1^{a_1}z_2^{a_2}\ldots z_n^{a_n}$.
    We use encodings of configurations to obtain an encoding of a run by concatenating encodings of its configurations.
    Finally, we define language $L_{M,l_0} = \{ w^r \mid w \text{ is an encoding of a run in $M$ from } l_0(0,0,\ldots, 0)\}$.
     
    \begin{claim}\label{clm:recognised}
    One can construct $1$-dimensional VASS recognising $\widebar{L_{M,l_0}}$.
    \end{claim}
    
    \begin{proof}
    \Wlog we assume, that there is at most one transition between each pair of states (otherwise we may split a state into several states,
    one for each incoming letter).
We construct $A$ such that it accepts $w$ if and only if $w^r$ does not represent a valid run of $M$ from $l_0(0,0,\ldots, 0)$.
    The idea is, that $A$ guesses the violation in the run represented by the word $w$. We have three types of violations. The first type is a control state violation, that is the run
    uses nonexisting transition or does not start in $l_0(0,0,\ldots,0)$. The second type is a counter violation, that we have invalid counter values between two consecutive configurations.   
    The third violation is that we have invalid encoding of a configuration, that is we have two consecutive letters $z_iz_j$ such that $j > i$.
    To spot control violation for nonexisting transitions we will have gadget for each pair of states $p$ and $q$ such that there is no transition from $p$ to $q$ spotting infix of the form $z_n^*z_{n-1}^*\ldots z_1^*qz_n^*z_{n-1}^*\ldots z_1^*p$. Such gadget is just an NFA.
    Observe, that if a run start in $l_0(0,0,\ldots, 0)$ then either the whole encoding of the run is equal to $l_0$ or suffix is of the form $pl_0$ for some $p \in Q$.
    Therefore to spot control violation, that the run does not start in $l_0(0,0,\ldots, 0)$ we will have an NFA  recognising words such that suffix is not of the form $pl_0$ for some $p \in Q$ and the encoding is not equal to $l_0$.
    To spot counter violation we will have a 1-dimensional VASS for each transition and each counter spotting violation when firing this transition on this counter.
    Let us fix transition $t$ and counter $z_i$. Let transition $t$ be from state $p$ to state $q$. We have four possibilities for the operation performed by $t$ on the counter $z_i$.
    Therefore we need to spot infix of the form:
    \begin{enumerate}
        \item $z_n^*z_{n-1}^*\ldots z_i^{a}\ldots z_1^*qz_n^*z_{n-1}^*\ldots z_i^{b}\ldots z_1^*p$ where $a > b - 1$ (equivalently $a \geq b$) if transition $t$ decrements counter $z_i$. 
        \item $z_n^*z_{n-1}^*\ldots z_i^{a}\ldots z_1^*qz_n^*z_{n-1}^*\ldots z_i^{b}\ldots z_1^*p$ where $a > b + 1$ if transition $t$ increments counter $z_i$. 
        \item $z_n^*z_{n-1}^*\ldots z_i^{a}\ldots z_1^*qz_n^*z_{n-1}^*\ldots z_i^{b}\ldots z_1^*p$ where $a > b$ if transition $t$ has no effect on the counter $z_i$. 
        \item $z_n^*z_{n-1}^*\ldots z_i^{a}\ldots z_1^*qz_n^*z_{n-1}^*\ldots z_i^{b}\ldots z_1^*p$ where $a > 0$ or $b > 0$ if transition $t$ zero-tests counter $z_i$. 
    \end{enumerate}
    All of the above can be done with $1$-dimensional VASS. To spot invalid encoding of a configuration for each $1 \leq i < j \leq n$  we will have an NFA recognising words with infix $z_iz_j$. 
    As all the gadgets have one initial state in which we "ignore" some prefix of the word, we can join them and obtain one dimensional VASS with single initial configuration.
    \end{proof}
    \begin{claim}\label{clm:regular}
    If $\reach(l_0)$ is finite then $\widebar{L_{M,l_0}}$ is a regular language.
    \end{claim}
    \begin{proof}
        Observe, that in the proof of Claim \ref{clm:recognised} only gadgets spotting counter violations were not an NFA. Therefore it is enough to replace them by some NFAs.
        As $\reach(l_0)$ is finite there exists a bound $B \in \N$ on the possible values of the counters. Hence we can implement each gadget spotting counter violation using the fact, that $B \geq a,b$. For instance for transition $t$ from state $p$ to $q$ having no effect on the counter $z_i$ we can have an NFA being a union of NFAs for each $0 \leq b < a \leq B$ spotting infix of the form  $z_n^*z_{n-1}^*\ldots z_i^{a}\ldots z_1^*qz_n^*z_{n-1}^*\ldots z_i^{b}\ldots z_1^*p$.
        In this way we will not detect counter violation increasing the counter to the value above $B$. Therefore we add another gadget spotting, that one counter is above $B$ hence then the word does not encode a correct run.
        This can be done by spotting for each $i \in [1,n]$ infix $z_i^{B+1}$, which can be done by an NFA.
        Because every gadget is now an NFA we showed that $\widebar{L_{M,l_0}}$ is a regular language when $\reach(l_0)$ is finite.
    \end{proof} 
    \begin{claim}\label{clm:not_k_amb}
    If $\reach(l_0)$ is infinite then $\widebar{L_{M,l_0}}$ is not recognised by a boundedly-ambiguous VASS.
    \end{claim}
    \begin{proof}
        Assume, towards contradiction, that $\widebar{L_{M,l_0}}$ is recognized by a \kamb VASS for some $k \in \N_+$. Recall, that there exist $q_1, q_2 \in Q$ such that from each $q_3 \in Q$ such that $q_3 \neq q_2$ and $q_3 \neq q_1$  there exists a transition from $q_3$ to $q_1$ with no effect on the counters. 
        Moreover, for each $i \in [2,n]$ there exists a transition from $q_1$ to $q_1$ decrementing the $i$th counter and increasing the first counter.
        Additionally, there is a single transition leaving $q_1$ decrementing the first counter to $q_2$ and there is exactly one transition leaving $q_2$, which goes to $q_2$ and decrements the first counter.
       
        As $k$-ambiguous VASS languages are closed under intersection with a regular language also $L_1 = \widebar{L_{M,l_0}} \cap (z_1^*q_2)^{k+2}Z^*q_1(Q \cup Z \setminus q_2)^*$ is recognised by a \kamb VASS.
        The idea of this intersection is to get a language similar to the language presented in Lemma \ref{lem:first_tool}.
        Observe, that each $w \in L_1$ can be uniquely decomposed into $w = z_1^{n_1}q_2z_1^{n_2}q_2\ldots z_1^{n_{k+2}}q_2v$. We denote $v$ by $\suff(w)$. We define $L = \setof{\suff(w)}{w \in L_1}$.
        We define function $f$ on words from $L$ by setting $f(w) = 0$ if $w^r$ encodes an incorrect run of $M$ from $l_0(0, 0, \ldots, 0)$ and otherwise $f(w) = n$ where $n$ is the maximal number such that $w = z_1^nv$ for some word $v$.
        Observe, that $L_1 = \{z_1^{n_1}q_2z_1^{n_2}q_2z_1^{n_3}q_2\ldots z_1^{n_{k+2}}q_2w \mid w \in L, \exists_{1 \leq i \leq k+1} n_i \geq n_{i+1} \lor n_{k+2} \geq f(w) \}$.
        This is because both transition from $q_1$ to $q_2$ and loop from $q_2$ to $q_2$ decrements the first counter and therefore for each $v \in L_1$ such that $v = z_1^{n_1}q_2z_1^{n_2}q_2z_1^{n_3}q_2\ldots z_1^{n_{k+2}}q_2w$  we have that $v^r$ encodes invalid run if and only if either $\suff(v)^r = w^r$ encodes an incorrect run (that is $f(w) = 0 \leq n_{k+2}$) or $v^r$ encodes a correct run but $f(w) \leq n_{k+2}$ (recall the definition of $f(w)$ in this case) or there exists $i \in [k+1]$ such that $n_i \geq n_{i+1}$.
        Moreover, notice that because $\reach(l_0)$ is infinite and we have transitions moving values to the first counter from all the other counters in the state $q_1$  for each each $B \in \N$ there exist $B' \in \N$ such that $B \leq B'$ and a correct run $\rho_{B}$ of $M$ from $l_0(0,0, \ldots, 0)$ to $q_1(B', 0, \ldots, 0)$. 
        Let $w_B$ be the encoding of this run. Observe, that $(z_1^{B'}q_2)^{k+2}w_B^r \in L_1$. Hence $w_B^r \in L$ and $f(w_B^r) = B' \geq B$.
        Hence $\sup(f) = \omega$ and hence, by Lemma \ref{lem:first_tool}, $L_1$ is not recognised by a $k$-ambiguous VASS. Contradiction.
    \end{proof}
% !TEX root = main.tex
\section{Deciding regularity}\label{sec:regularity}
In this section we present a proof of the following theorem, that deciding regularity of a language of a $k$-ambigous VASS is decidable. This is in contrast to the general case where regularity is undecidable \cite{regularity}.
\begin{theorem}\label{thm:regularity}
    For every $k \in \N$ it is decidable whether for a given $k$-ambigous VASS $A$ language $L(A)$ is regular.
\end{theorem}
The proof uses two, already present in the literature, results.
The first result is the following Theorem from \cite{DBLP:conf/concur/CzerwinskiH22}.
\begin{theorem}[Theorem 28 \cite{DBLP:conf/concur/CzerwinskiH22}]\label{thm:complement}
    For any $k \in \N$ and a k-ambiguous VASS one can build in elementary time
a downward-VASS which recognises the complement of its language.
\end{theorem}
The second result is the decidability of regular separability of coverability VASS language and reachability VASS language.
Regular separability problem asks whether for two VASSs $A$ and $B$ there exists a regular language $L$ such that $L(A) \subseteq L$ and $L(B) \,\cap\, L = \emptyset$.
\begin{theorem}[Theorem 7 \cite{DBLP:conf/lics/CzerwinskiZ20}]\label{thm:reg_sep}
    Regular separability is decidable if one VASS is a coverability VASS and the second VASS is a reachability VASS.
    \footnote{Recently even stronger result about decidability of regular separability for reachability VASSs was proven in \cite{DBLP:conf/lics/Keskin024}.}
\end{theorem}

We devote the rest of this section to the proof of Theorem \ref{thm:regularity}. Let us fix a $k$-ambiguous VASS $V$. 
We show how to decide regularity of a language of a $k$-ambiguous VASS.
Using Theorem \ref{thm:complement} we get a downward-VASS $\hat{V}$ recognising the complement of $L(V)$. Observe, that $L(V)$ is regular if and only if $L(V)$ and $L(\hat{V})$ are regular separable.
Now we will use the following claim about downward-VASSs.
\begin{claim}\label{clm:downward_to_reachability}
    For every downward-VASS one can construct reachability VASS recognising the same language.
\end{claim}
\begin{proof}
    Let us fix $d$-dimensional downward-VASS $V$ and let $Q$ be the states of $V$. We know, that the set of accepting configurations $F$ is downward-closed. Hence, due to Proposition \ref{prop:down_atom}, we have a finite set
    $D \subseteq Q \times (\N \cup \omega)^d$ such that:
    $$F = \bigcup_{q(u) \in D} \set{q} \times u\downarrow $$
    We obtain reachability VASS $V'$ recognising the same language as $V$ by taking VASS $V$ and for each $q(u) \in D$
    adding state $q_{q(u)}$, $\varepsilon$ transition from with no effect from $q$ to $q_{q(u)}$. 
    Moreover, for each $i \in [1,d]$ if $u_i \in \N$ we add $\varepsilon$ transition from $q_{q(u)}$ to $q_{q(u)}$ incrementing by one the $i$th counter and otherwise if $u_i = \omega$ we add an $\varepsilon$-transition from $q_{q(u)}$ to $q_{q(u)}$ decrementing by one $i$th counter.
    Let the set of added states be equal to $Q'$. We set the initial configurations of $V'$ to be the initial conditions of $V$. For vector $u \in (\N \cup \omega)^d$ let us denote by $\hat{u} \in \N^d$ vector such that for all $i \in [1,d]$ we have $u_i = \hat{u}_i$ if $u_i \in \N$ and $\hat{u}_i = 0$ otherwise. We set the accepting configurations of $V'$ to the following set of configurations $F' = \setof{q_{q(u)}(\hat{u})}{q_{q(u)} \in Q'}$. 
    Now we have to prove, that $L(V) = L(V')$. For $L(V) \subseteq L(V')$ observe, that for each $w \in L(V)$ there exists configuration $q(v) \in F$ such that $w$ is read by an accepting run $\rho$ ending in $q(v)$.
    As $F$ is a downward-closed set there is $u$ such that $v \preceq u$ such that $q_{q(v)}(\hat{v}) \in F'$. Moreover, observe that $q_{q(v)}(\hat{v})$ is reachable from $q(u)$ in $V$ using $\varepsilon$-transitions.
    Hence $w \in L(V')$.

    For $L(V') \subseteq L(V)$ observe that for each $w \in L(V')$ we have an accepting run $\rho$ ending in configuration $q_{q(u)}(\hat{u}) \in F'$. Due to construction of $V'$ we know, that there exists prefix of $\rho$, such that it also reads $w$ and reaches configuration $q(v)$ such that $v \preceq u$ and hence $q(v) \in F$ and $w \in L(V)$.
\end{proof}
Now we invoke Claim \ref{clm:downward_to_reachability} on $\hat{V}$ to get reachability VASS $\hat{V'}$ recognising the same language.
Finally, we use algorithm from Theorem \ref{thm:reg_sep} for $V$ and $\hat{V'}$ knowing that they are regular separable if and only if $L(V)$ is regular. This concludes the proof of Theorem \ref{thm:regularity}.

% !TEX root = main.tex
\section{Future research}\label{sec:future}

\myparagraph{Other problems for boundedly-ambiguous VASSs}
We have shown in Section~\ref{sec:regularity} that the regularity problem is decidable for baVASSs (boundedly-ambiguous VASSs),
in contrast to general VASSs. It is also known that the language equivalence problem is decidable for baVASSs~\cite{DBLP:journals/lmcs/CzerwinskiH25}, while being undecidable for general VASSs~\cite{DBLP:journals/tcs/ArakiK76,DBLP:conf/lics/HofmanMT13}.
It is natural to ask whether other classical problems are decidable for baVASSs, for example deciding whether there exists
an equivalent deterministic or unambiguous VASS. These problems are undecidable for general VASS due to Theorem~\ref{thm:main},
but can possibly be decidable for baVASSs. One can also ask whether given $k$-ambiguous VASS has an equivalent $(k-1)$-ambiguous VASS, our techniques from Section~\ref{sec:undecidability} used for showing undecidability does not seem to work in that case.
Other further research for baVASSs would be to ask about complexity of the mentioned problems, for example to understand
complexity of the regularity problem for baVASSs, since we already know it is decidable.

\myparagraph{Languages of VASSs accepting by configuration}
In this paper we have considered VASSs accepting by set of accepting states, it is natural to ask
what happens if we modify the acceptance set to be a finite set of accepting configurations
or even a single configuration.
For VASSs accepting by configuration already the language universality problem is undecidable.
We are not aware of any citation, but the proof is rather straightforward and uses the classical technique.
For a given two-counter automaton one can construct $1$-VASS accepting by configuration, which recognises
all the words beside correct encodings of the runs of the two-counter automaton. Therefore the reachability problem
for two-counter automaton, which is undecidable, can be reduced to the universality problem for $1$-VASSs accepting by a configuration.
Since the universality problem is undecidable for VASSs accepting by configuration, there is not much hope that the other nontrivial
problems (beside emptiness) are decidable.

However, one can ask what about VASSs accepting by configuration with some restriction on nondeterminism, for
example unambiguous or boundedly-ambiguous VASSs accepting by configuration. It is natural to ask
whether universality, language equivalence, regularity or determinisability problems are decidable for that models
and what is its complexity.

\section{Acknowledgements}
We would like to thank David Purser for sharing with us a conjecture about language not recognised by an unambiguous VASS.
We thank also Sławomir Lasota for suggesting us to reduce from 0-finite-reach instead of reducing from the halting problem for 2CM in the proof of Theorem \ref{thm:main}.
Additionally, we would like to thank Piotr Hofman for inspiring discussions and for reviewing the Master's Thesis \cite{Orlikowski2024}, which served as the foundation for the results presented in this paper.

%%
%% Bibliography
%%

%% Please use bibtex, 

\newpage

\bibliography{citat}

\begin{thebibliography}{10}

\bibitem{DBLP:conf/concur/AlmagorY22}
Shaull Almagor and Asaf Yeshurun.
\newblock {Determinization of One-Counter Nets}.
\newblock In {\em Proceedings of {CONCUR} 2022}, volume 243 of {\em LIPIcs}, pages 18:1--18:23. Schloss Dagstuhl - Leibniz-Zentrum f{\"{u}}r Informatik, 2022.
\newblock \href {https://doi.org/10.4230/LIPICS.CONCUR.2022.18} {\path{doi:10.4230/LIPICS.CONCUR.2022.18}}.

\bibitem{DBLP:journals/tcs/ArakiK76}
Toshiro Araki and Tadao Kasami.
\newblock {Some Decision Problems Related to the Reachability Problem for Petri Nets}.
\newblock {\em Theor. Comput. Sci.}, 3(1):85--104, 1976.
\newblock \href {https://doi.org/10.1016/0304-3975(76)90067-0} {\path{doi:10.1016/0304-3975(76)90067-0}}.

\bibitem{DBLP:conf/stacs/BarloyC21}
Corentin Barloy and Lorenzo Clemente.
\newblock Bidimensional linear recursive sequences and universality of unambiguous register automata.
\newblock In {\em Proceedings of {STACS} 2021)}, volume 187 of {\em LIPIcs}, pages 8:1--8:15. Schloss Dagstuhl - Leibniz-Zentrum f{\"{u}}r Informatik, 2021.

\bibitem{DBLP:conf/lics/BellS23}
Jason~P. Bell and Daniel Smertnig.
\newblock {Computing the linear hull: Deciding Deterministic? and Unambiguous? for weighted automata over fields}.
\newblock In {\em Proceedings of {LICS} 2023}, pages 1--13. {IEEE}, 2023.
\newblock \href {https://doi.org/10.1109/LICS56636.2023.10175691} {\path{doi:10.1109/LICS56636.2023.10175691}}.

\bibitem{DBLP:journals/theoretics/BojanczykFKM24}
Mikolaj Bojanczyk, Joanna Fijalkow, Bartek Klin, and Joshua Moerman.
\newblock {Orbit-Finite-Dimensional Vector Spaces and Weighted Register Automata}.
\newblock {\em TheoretiCS}, 3, 2024.
\newblock \href {https://doi.org/10.46298/THEORETICS.24.13} {\path{doi:10.46298/THEORETICS.24.13}}.

\bibitem{DBLP:conf/lics/BojanczykKM21}
Mikolaj Bojanczyk, Bartek Klin, and Joshua Moerman.
\newblock Orbit-finite-dimensional vector spaces and weighted register automata.
\newblock In {\em Proceedings of {LICS} 2021}, pages 1--13. {IEEE}, 2021.

\bibitem{history-deterministic}
Sougata Bose, David Purser, and Patrick Totzke.
\newblock {History-Deterministic Vector Addition Systems}.
\newblock In {\em Proceedings of {CONCUR} 2023}, volume 279 of {\em LIPIcs}, pages 18:1--18:17, 2023.

\bibitem{DBLP:conf/dcfs/Colcombet15}
Thomas Colcombet.
\newblock Unambiguity in automata theory.
\newblock In {\em Proceedings of {DCFS} 2015}, pages 3--18, 2015.

\bibitem{DBLP:conf/concur/CzerwinskiFH20}
Wojciech Czerwinski, Diego Figueira, and Piotr Hofman.
\newblock {Universality Problem for Unambiguous VASS}.
\newblock In {\em Proceedings of {CONCUR} 2020}, volume 171 of {\em LIPIcs}, pages 36:1--36:15. Schloss Dagstuhl - Leibniz-Zentrum f{\"{u}}r Informatik, 2020.
\newblock \href {https://doi.org/10.4230/LIPICS.CONCUR.2020.36} {\path{doi:10.4230/LIPICS.CONCUR.2020.36}}.

\bibitem{DBLP:conf/concur/CzerwinskiH22}
Wojciech Czerwinski and Piotr Hofman.
\newblock {Language Inclusion for Boundedly-Ambiguous Vector Addition Systems Is Decidable}.
\newblock In {\em Proceedings of {CONCUR} 2022}, volume 243 of {\em LIPIcs}, pages 16:1--16:22. Schloss Dagstuhl - Leibniz-Zentrum f{\"{u}}r Informatik, 2022.

\bibitem{DBLP:journals/lmcs/CzerwinskiH25}
Wojciech Czerwinski and Piotr Hofman.
\newblock {Language Inclusion for Boundedly-Ambiguous Vector Addition Systems is Decidable}.
\newblock {\em Log. Methods Comput. Sci.}, 21(1), 2025.
\newblock \href {https://doi.org/10.46298/LMCS-21(1:27)2025} {\path{doi:10.46298/LMCS-21(1:27)2025}}.

\bibitem{DBLP:journals/lmcs/CzerwinskiL19}
Wojciech Czerwinski and Slawomir Lasota.
\newblock {Regular Separability of One Counter Automata}.
\newblock {\em Log. Methods Comput. Sci.}, 15(2), 2019.
\newblock \href {https://doi.org/10.23638/LMCS-15(2:20)2019} {\path{doi:10.23638/LMCS-15(2:20)2019}}.

\bibitem{CzerwinskiLMMKS18}
Wojciech Czerwi{\'n}ski, S{\l}awomir Lasota, Roland Meyer, Sebastian Muskalla, K.~Narayan Kumar, and Prakash Saivasan.
\newblock Regular separability of well-structured transition systems.
\newblock In {\em 29th International Conference on Concurrency Theory, {CONCUR} 2018}, volume 118 of {\em LIPIcs}, pages 35:1--35:18. Schloss Dagstuhl - Leibniz-Zentrum f{\"{u}}r Informatik, 2018.
\newblock \href {https://doi.org/10.4230/LIPIcs.CONCUR.2018.35} {\path{doi:10.4230/LIPIcs.CONCUR.2018.35}}.

\bibitem{DBLP:conf/icalp/CzerwinskiMQ21}
Wojciech Czerwinski, Antoine Mottet, and Karin Quaas.
\newblock New techniques for universality in unambiguous register automata.
\newblock In {\em Proceedings of {ICALP} 2021}, volume 198 of {\em LIPIcs}, pages 129:1--129:16. Schloss Dagstuhl - Leibniz-Zentrum f{\"{u}}r Informatik, 2021.

\bibitem{DBLP:conf/lics/CzerwinskiZ20}
Wojciech Czerwinski and Georg Zetzsche.
\newblock {An Approach to Regular Separability in Vector Addition Systems}.
\newblock In {\em {LICS} '20: 35th Annual {ACM/IEEE} Symposium on Logic in Computer Science, Saarbr{\"{u}}cken, Germany, July 8-11, 2020}, pages 341--354. {ACM}, 2020.

\bibitem{Dickson}
L.E. Dickson.
\newblock Finiteness of the odd perfect and primitive abundant numbers with n distinct prime factors.
\newblock {\em American Journal of Mathematics}, 35((4)):413–422, 1913.

\bibitem{undecidable_cfl}
Seymour Ginsburg and Joseph Ullian.
\newblock {Ambiguity in Context Free Languages}.
\newblock {\em J. ACM}, 13(1):62–89, 1966.
\newblock \href {https://doi.org/10.1145/321312.321318} {\path{doi:10.1145/321312.321318}}.

\bibitem{DBLP:conf/icalp/GoosK022}
Mika G{\"{o}}{\"{o}}s, Stefan Kiefer, and Weiqiang Yuan.
\newblock {Lower Bounds for Unambiguous Automata via Communication Complexity}.
\newblock In {\em Proceedings of {ICALP} 2022}, volume 229 of {\em LIPIcs}, pages 126:1--126:13. Schloss Dagstuhl - Leibniz-Zentrum f{\"{u}}r Informatik, 2022.
\newblock \href {https://doi.org/10.4230/LIPICS.ICALP.2022.126} {\path{doi:10.4230/LIPICS.ICALP.2022.126}}.

\bibitem{DBLP:journals/mst/Greibach68}
Sheila~A. Greibach.
\newblock {A Note on Undecidable Properties of Formal Languages}.
\newblock {\em Math. Syst. Theory}, 2(1):1--6, 1968.
\newblock \href {https://doi.org/10.1007/BF01691341} {\path{doi:10.1007/BF01691341}}.

\bibitem{semilinear}
Christoph Haase.
\newblock {A survival Guide to Presburger Arithmetic}.
\newblock {\em {ACM} {SIGLOG} News}, 5(3):67--82, 2018.

\bibitem{DBLP:conf/lics/HofmanMT13}
Piotr Hofman, Richard Mayr, and Patrick Totzke.
\newblock Decidability of weak simulation on one-counter nets.
\newblock In {\em Proceedings of {LICS} 2013}, pages 203--212. {IEEE} Computer Society, 2013.

\bibitem{coverability_tree}
Richard~M. Karp and Raymond~E. Miller.
\newblock Parallel program schemata.
\newblock {\em J. Comput. Syst. Sci.}, 3(2):147--195, 1969.
\newblock \href {https://doi.org/10.1016/S0022-0000(69)80011-5} {\path{doi:10.1016/S0022-0000(69)80011-5}}.

\bibitem{DBLP:conf/lics/Keskin024}
Eren Keskin and Roland Meyer.
\newblock On the separability problem of {VASS} reachability languages.
\newblock In {\em Proceedings of the 39th Annual {ACM/IEEE} Symposium on Logic in Computer Science, {LICS} 2024, Tallinn, Estonia, July 8-11, 2024}, pages 49:1--49:14. {ACM}, 2024.

\bibitem{konig}
D{\'e}nes K{\"o}nig.
\newblock {\"U}ber eine schlussweise aus dem endlichen ins unendliche.
\newblock {\em Acta Sci. Math.(Szeged)}, 3(2-3):121--130, 1927.

\bibitem{lcm1}
Richard Mayr.
\newblock Undecidable problems in unreliable computations.
\newblock {\em Theor. Comput. Sci.}, 297(1-3):337--354, 2003.

\bibitem{DBLP:conf/stacs/MottetQ19}
Antoine Mottet and Karin Quaas.
\newblock The containment problem for unambiguous register automata.
\newblock In {\em Proceedings of {STACS} 2019}, pages 53:1--53:15, 2019.

\bibitem{DBLP:journals/tocl/NevenSV04}
Frank Neven, Thomas Schwentick, and Victor Vianu.
\newblock Finite state machines for strings over infinite alphabets.
\newblock {\em {ACM} Trans. Comput. Log.}, 5(3):403--435, 2004.

\bibitem{DBLP:conf/fossacs/PrakashT23}
Aditya Prakash and K.~S. Thejaswini.
\newblock {On History-Deterministic One-Counter Nets}.
\newblock In {\em Proceedings of {FoSSaCS} 2023}, volume 13992 of {\em Lecture Notes in Computer Science}, pages 218--239. Springer, 2023.
\newblock \href {https://doi.org/10.1007/978-3-031-30829-1\_11} {\path{doi:10.1007/978-3-031-30829-1\_11}}.

\bibitem{DBLP:journals/corr/abs-2410-03444}
Antoni Puch and Daniel Smertnig.
\newblock {Factoring through monomial representations: arithmetic characterizations and ambiguity of weighted automata}.
\newblock {\em CoRR}, abs/2410.03444, 2024.
\newblock \href {https://doi.org/10.48550/ARXIV.2410.03444} {\path{doi:10.48550/ARXIV.2410.03444}}.

\bibitem{DBLP:conf/icalp/Raskin18}
Mikhail Raskin.
\newblock A superpolynomial lower bound for the size of non-deterministic complement of an unambiguous automaton.
\newblock In {\em Proceedings of {ICALP} 2018}, pages 138:1--138:11, 2018.

\bibitem{DBLP:journals/siglog/Schmitz16}
Sylvain Schmitz.
\newblock {The complexity of reachability in vector addition systems}.
\newblock {\em {ACM} {SIGLOG} News}, 3(1):4--21, 2016.

\bibitem{schmitz2012algorithmic}
Sylvain Schmitz and Philippe Schnoebelen.
\newblock {Algorithmic Aspects of WQO Theory}.
\newblock {\em \url{https://cel.hal.science/cel-00727025/document}}, 2012.

\bibitem{lcm2}
Philippe Schnoebelen.
\newblock {Lossy Counter Machines Decidability Cheat Sheet}.
\newblock In {\em Reachability Problems, 4th International Workshop, {RP} 2010, Brno, Czech Republic, August 28-29, 2010. Proceedings}, volume 6227 of {\em Lecture Notes in Computer Science}, pages 51--75. Springer, 2010.

\bibitem{DBLP:journals/siamcomp/StearnsH85}
Richard~Edwin Stearns and Harry B.~Hunt III.
\newblock On the equivalence and containment problems for unambiguous regular expressions, regular grammars and finite automata.
\newblock {\em {SIAM} J. Comput.}, 14(3):598--611, 1985.

\bibitem{DBLP:journals/ipl/Tzeng96}
Wen{-}Guey Tzeng.
\newblock On path equivalence of nondeterministic finite automata.
\newblock {\em Inf. Process. Lett.}, 58(1):43--46, 1996.

\bibitem{regularity}
R{\"{u}}diger Valk and Guy Vidal{-}Naquet.
\newblock Petri nets and regular languages.
\newblock {\em J. Comput. Syst. Sci.}, 23(3):299--325, 1981.

\bibitem{Orlikowski2024}
Łukasz Orlikowski.
\newblock {Languages of Unambiguous Vector Addition Systems with States}.
\newblock Master's thesis, University of Warsaw, 2024.
\newblock URL: \url{https://apd.uw.edu.pl/diplomas/225085/}.

\end{thebibliography}

\newpage

%\appendix

%\appendix

\end{document}